\documentclass[journal,onecolumn]{IEEEtran}
\IEEEoverridecommandlockouts
\usepackage[T1]{fontenc} 
\usepackage{times}
\usepackage[cmex10]{amsmath}
\usepackage{amssymb}
\usepackage{amsthm}
\usepackage{color}
\usepackage{xcolor}
\usepackage{algorithm}
\usepackage[scr=rsfs]{mathalpha}
\usepackage{graphicx}
\usepackage[caption=false]{subfig}
\usepackage{multirow}
\usepackage{algpseudocode}
\usepackage[bookmarks=false,colorlinks=false,pdfborder={0 0 0}]{hyperref}
\usepackage{cite}
\usepackage{bm}
\usepackage{arydshln}
\usepackage{mathtools}
\usepackage{microtype}
\usepackage[figuresright]{rotating}
\usepackage{threeparttable}
\usepackage{booktabs}
\usepackage{verbatim}

\usepackage{amsthm} 

\newtheorem{construction}{Construction}
\newtheorem{example}{Example}
\theoremstyle{definition}
\newtheorem{definition}{Definition}
\usepackage{enumitem}

\newtheorem{theorem}{Theorem}

\newtheorem{lemma}[theorem]{Lemma}

\def\BibTeX{{\rm B\kern-.05em{\sc i\kern-.025em b}\kern-.08em
		T\kern-.1667em\lower.7ex\hbox{E}\kern-.125emX}}

\long\def\symbolfootnote[#1]#2{\begingroup
	\def\thefootnote{\fnsymbol{footnote}}\footnote[#1]{#2}\endgroup}
\IEEEoverridecommandlockouts

\title{Asymptotically Optimal Codes for Correcting Burst Deletions and Insertions in Labeled DNA Sequences}

\author{Wenhao Liu$^\S$, Zhengyi Jiang$^\S$, Zhongyi Huang$^\S$, Hanxu Hou$^\dagger$\\
	$^\S$ Department of Mathematical Sciences, Tsinghua University\\
	$^\dagger$ Shenzhen University of Advanced Technology
}

\begin{document}
	\let\emph\textit
	\maketitle
	\pagestyle{empty}  
	\thispagestyle{empty} 
	\begin{abstract}
        Fluorescent labeling is a cornerstone of DNA visualization and a key enabler of random access in DNA-based data storage. However, the stochastic nature of biochemical processes, including synthesis, hybridization, and optical readout, induces \emph{burst} synchronization errors within the resulting labeling sequences. To address this critical challenge, we formally introduce \emph{burst $t$-deletion/insertion $\mathcal{A}$-labeling codes,} 
        designed to correct a single burst of $t$  deletions or insertions in the label domain. Our contributions are threefold.
		\begin{itemize}
			\item \textbf{Fundamental limit.} We establish an information-theoretic lower bound of $\log_4 n + \mathcal{O}(1)$ on the redundancy of any such code for all $t \ge 1$ with $t \mid n$. To the best of our knowledge, this resolves the first information-theoretic lower bound even for the single-error case \(t=1\).
			\item \textbf{Explicit construction.} For $t \ge 2$, $t \mid n$, and $n \ge 7t + 3$, we propose explicit encoding and decoding algorithms, both running in $\mathcal{O}(n^2)$ time. A novel generalized Run-Length Limited (RLL) constraint is introduced to bridge the structural mismatch between the  DNA encoding domain and the label error domain.
			\item \textbf{Asymptotic optimality.} The proposed scheme achieves redundancy $\log_4 n + (t-1)\log_4 \log_{8/3} n + \mathcal{O}(1)$, matching the dominant term of the lower bound up to a small $\mathcal{O}(\log\log n)$ overhead, rendering the construction asymptotically optimal for fixed $t$.
		\end{itemize}
	\end{abstract}
	
	
	\section{Introduction}
	DNA has attracted significant attention as an ultra-high-density and durable medium for data storage, yet efficient and accurate data retrieval remains a fundamental bottleneck. Conventional sequencing technologies, such as Next-Generation Sequencing and Nanopore sequencing~\cite{church2012next, goldman2013towards, lopez2019dna}, are well established but suffer from high cost and limited throughput. As an alternative, affinity-based fluorescent labeling has emerged as a promising low-cost readout mechanism~\cite{levy2013beyond, jeffet2021single}.
	
	Fluorescent labeling has been a cornerstone of molecular biology for decades, supporting visualization and targeted sequence identification through techniques such as Fluorescence \emph{in situ} Hybridization (FISH)~\cite{moter2000fluorescence} and optical mapping~\cite{muller2017optical}. Adapting these mature tools to DNA-based data storage avoids the cost of resolving every base individually: biochemical labels bind to short target patterns, and data is recovered by optically locating these labels along the strand.

    Recently, Hanania \textit{et al.}~\cite{2025ISIT, 2025TIT} formalized this probing process as a communication channel and analyzed its information-theoretic capacity. Their model makes explicit a fundamental ambiguity of the channel: each label reveals only the \emph{positions} at which it appears. For example, given the underlying sequence $\bm{x}=\mathrm{ACGACTA}$ and a single label $\bm{\alpha}=\mathrm{AC}$, the label matches starting at the 1st and 4th positions. The resulting labeling output $L_{\bm{\alpha}}(\bm{x})=1001000$ explicitly indicates these match positions with $1$s and the mismatches with $0$s. This sequence records matches and mismatches but does not specify which base sits at any non-matching position. Hence, a single label is generally insufficient for unique reconstruction.
	
	To overcome this ambiguity, a \emph{multi-label} probing scheme is essential. Hanania \textit{et al.}~\cite{2025TIT} proved that, among all label sets consisting of length-two patterns, a specific size-$10$ set
	\begin{equation*}
		\mathcal{S} = \{\mathrm{AC, CA, GA, GC, GG, GT, TA, TC, TG, TT}\}
	\end{equation*}
	is the minimum-cardinality set for which the original DNA sequence can be uniquely reconstructed from the joint labeling output, provided the two boundary bases are known. Let $\Sigma_q^n$ denote the set of length-$n$ sequences over a $q$-ary alphabet $\Sigma_q$. Throughout the introduction, we informally write $L_{\mathcal{S}}(\bm{x}) \in \Sigma_{11}^n$ (where $\Sigma_{11}=\{0,1,\ldots,10\}$) for the resulting \emph{labeling sequence} of a DNA strand $\bm{x}\in\Sigma_4^n$ (where $\Sigma_4$ represents the four DNA nucleotides): its $i$-th entry equals $k\in\{1,\ldots,10\}$ if the $i$-th dinucleotide of $\bm{x}$ is the $k$-th label of $\mathcal{S}$, and $0$ otherwise.
	
	Although early models assumed an error-free readout, practical biochemical labeling is intrinsically noisy. Stochastic hybridization, photophysics, and limited optical resolution all produce synchronization-type errors, predominantly insertions and deletions (indels)~\cite{muller2017optical, 2025ISIT}: missed labels, off-target labels, or unreadable spots shift the positional indices and severely degrade reconstruction. To handle this, Hanania \textit{et al.}~\cite{2025ISIT} proposed an error-correcting framework over $\mathcal{S}$ that corrects a \emph{single} label-domain indel. In their scheme the user data $\bm{u}$ is encoded into a constrained DNA strand $\bm{x}$, and the (possibly corrupted) labeling output is then decoded back to $\bm{u}$.
	
	Single-error model is, however, often inadequate in practice. Many physical artefacts of optical labeling produce \emph{bursts} of consecutive errors~\cite{schoeny2017codes, gabrys2017asymmetric}: continuous fluorophore photobleaching, regional binding failures, and localized resolution loss can erase or distort several adjacent label positions at once~\cite{lam2012genome, levy2013beyond}. This motivates the design of codes that correct a burst of $t$ consecutive label-domain indels with $t \ge 2$, which is the focus of the present work.
	
	\subsection{Technical Challenges}
	Although burst-deletion codes for ordinary $q$-ary alphabets have been studied extensively~\cite{schoeny2017codes, 2024DVT, 2023bound}, their constructions and bounds \emph{do not} carry over to the label-domain channel in any direct way. The difficulty stems from a structural mismatch between the encoding domain and the error domain:
	\begin{enumerate}
		\item \textbf{Mismatch between the encoder and channel domains.}
		The encoder may freely choose the underlying DNA sequence $\bm{x}\in\Sigma_4^n$, but the error acts on the labeling sequence $L_{\mathcal{S}}(\bm{x})\in\Sigma_{11}^n$. The map $L_{\mathcal{S}}$ is many-to-one in general (only invertible given boundary bases) and \emph{non-surjective}: most sequences in $\Sigma_{11}^n$ are not valid labelings of any DNA strand. Standard $11$-ary burst-deletion codes (i.e., conventional error-correcting codes designed to correct consecutive deletions over an 11-symbol alphabet) cannot be used because their codewords might not lie in the image of $L_{\mathcal{S}}$.
		\item \textbf{Strong inter-symbol dependence.}
		Adjacent labels share an underlying base, so the labeling sequence inherits sharp local dependencies that ordinary i.i.d.-type counting arguments fail to capture. A single burst of insertions or deletions in the label domain corresponds, via $L_{\mathcal{S}}^{-1}$, to a structurally constrained perturbation of $\bm{x}$; the size of the resulting labeling ball (i.e., the set of all possible corrupted labeling sequences generated from $\bm{x}$) depends in a non-trivial way on $\bm{x}$ itself.
		\item \textbf{No matching lower bound.}
		Even for the single-error case $t=1$, the information-theoretic redundancy lower bound for this channel was previously unknown, which makes it impossible to certify the optimality of any explicit construction.
	\end{enumerate}
    These issues mean that even formulating an appropriate code—let alone proving optimality—requires techniques specific to the labeling channel. To bridge the structural gap between the encoder and the error channel, we propose a novel ``generalized RLL'' constraint. This constraint is applied directly to the DNA domain but is carefully designed to deterministically enforce a traditional run-length limit on the resulting labeling sequence, which is crucial for bounding the positional ambiguity of indels. Our explicit matrix construction then couples this generalized RLL constraint with label-domain Varshamov--Tenengolts (VT)~\cite{1984VT_def} and shifted VT (SVT)~\cite{2017QVT_def} redundancy, effectively reducing a length-$t$ burst error into a set of localizable single-row indels. Furthermore, to certify the theoretical efficiency of this scheme, our redundancy lower bound is established via a sphere-packing argument that operates strictly over the valid image of $L_{\mathcal{S}}$, completely circumventing the non-surjectivity of the labeling map.
	
	\subsection{Contributions}
	We consider the end-to-end DNA storage pipeline of encoding, synthesis, label-based readout, and decoding. Our main contributions are:
	\begin{enumerate}
		\item \textbf{Code model.} We formally define burst $t$-deletion/insertion $\mathcal{A}$-labeling codes over the minimal label set $\mathcal{S}$. To the best of our knowledge, this is the first coding model in this channel that handles burst indels of length $t \ge 2$.
		\item \textbf{Fundamental limit.} We derive the first information-theoretic lower bound on redundancy, establishing that the redundancy is at least $\log_4 n + \mathcal{O}(1)$ for any burst $t$-deletion/insertion $\mathcal{A}$-labeling codes. This resolves an open problem that was previously unsettled even for $t=1$.
        \item \textbf{Explicit construction.} For any burst length $t \ge 2$ and code length $n$ satisfying $t \mid n$ and $n \ge 7t + 3$, we propose an explicit code construction along with the corresponding encoding and decoding algorithms, both running in $\mathcal{O}(n^2)$ time, hence practical at realistic block lengths.
		\item \textbf{Asymptotic optimality.} The proposed scheme achieves redundancy
		\begin{equation*}
			r(n,t) = \log_4 n + (t-1)\log_4 \log_{8/3} n + \mathcal{O}(1),
		\end{equation*}
		which matches the dominant term of the lower bound up to an $\mathcal{O}(\log\log n)$ gap.
	\end{enumerate}
	
	\section{Preliminaries and System Model}
	\label{sec:preliminaries}
	
	This section fixes notation, formally introduces the label mapping, and reviews the classical building blocks used in our construction.
	
	\subsection{Notation and Label Mapping}
	
	Let $q \ge 2$ and $n$ be positive integers. We denote the $q$-ary alphabet by $\Sigma_q = \{0,1,\ldots,q-1\}$. We write $\Sigma_q^n$ for the set of length-$n$ sequences over $\Sigma_q$. For DNA storage, we identify the four nucleotides with $\Sigma_4$ via the bijection $\mathrm{A}\mapsto 0,\, \mathrm{T}\mapsto 1,\, \mathrm{C}\mapsto 2,\, \mathrm{G}\mapsto 3$.
	 
	We denote the cardinality of a set $S$ by $|S|$, the ring of integers modulo $n$ by $\mathbb{Z}_n$, and the integer interval $\{i,i+1,\ldots,j\}$ by $[i:j]$ for $i\le j$. For a vector $\bm{x}=(x_1,\ldots,x_n)\in\Sigma_q^n$, we write $\bm{x}_{[i_1:i_2]} \triangleq (x_{i_1},\ldots,x_{i_2})$. Matrix indexing follows standard conventions: For a matrix $\bm{X}$, $\bm{X}_{i,j}$ denotes the $(i,j)$-th entry; $\bm{X}_{i,[j_1:j_2]}$ and $\bm{X}_{[i_1:i_2],j}$ are the corresponding row and column segments; and $\bm{X}_{[i_1:i_2],[j_1:j_2]}$ is the submatrix indexed by these rows and columns. All logarithms in asymptotic statements are base $e$ (natural log) unless otherwise specified. Furthermore, a \emph{run} in a sequence $\bm{s}\in\Sigma_q^n$ is defined as a maximal contiguous subsegment of identical symbols. Formally, a substring $\bm{s}_{[i:i+r-1]}$ is a run of length $r\ge 1$ if $s_i=\cdots=s_{i+r-1}$, with the boundary conditions that $s_{i-1}\neq s_i$ (when $i>1$) and $s_{i+r-1}\neq s_{i+r}$ (when $i+r-1<n$).
	
	We now recall the formalism of labels and labeling sequences from~\cite{2025TIT, 2025ISIT}.
	
	\begin{definition}[Labels and Labeling Sequences \cite{2025TIT, 2025ISIT}]
		A \emph{label} is any sequence $\bm{\alpha}\in\Sigma_4^\ell$. Given a set of $k$ labels $\mathcal{A}=\{\bm{\alpha}_1,\ldots,\bm{\alpha}_k\}$ with $|\bm{\alpha}_j| = \ell_j$, and assuming that no label is a substring of another, the \emph{$\mathcal{A}$-labeling sequence} of $\bm{x}=(x_1,\ldots,x_n)\in\Sigma_4^n$ is the sequence $L_{\mathcal{A}}(\bm{x}) = (c_1,\ldots,c_n)\in\Sigma_{k+1}^n$ with
		\begin{equation*}
			c_i = \begin{cases}
				j, & \text{if } \exists\, j\in[1:k] \text{ such that } 
				 i \le n-\ell_j+1 \text{ and } \bm{x}_{[i:i+\ell_j-1]} = \bm{\alpha}_j, \\
				0, & \text{otherwise}.
			\end{cases}
		\end{equation*}
		Here, $c_i=0$ indicates a mismatch at position $i$. We assume no label is a substring of another to ensure unambiguous labeling. 
	\end{definition}
	
	Throughout the paper, we adopt the mini-cardinality label set $\mathcal{S}$ identified in \cite{2025TIT, 2025ISIT}.
	
	\begin{definition}[Minimal Label Set $\mathcal{S}$ \cite{2025TIT}]
		Let
		\begin{equation*}
			\mathcal{S} = \{\mathrm{AC, CA, GA, GC, GG, GT, TA, TC, TG, TT}\},
		\end{equation*}
		enumerated as $\bm{\alpha}_1,\ldots,\bm{\alpha}_{10}$. Consequently,  $L_{\mathcal{S}}(\bm{x})\in\Sigma_{11}^n$, where $0$ denotes a mismatch.
	\end{definition}
	
	The set $\mathcal{S}$ is distinguished by the following reconstruction property.
	
	\begin{theorem}[Unique Reconstruction and Minimality \cite{2025TIT,2025ISIT}]
		\label{thm:reconstruction}
		For any $\bm{x}=(x_1,\ldots,x_n)\in\Sigma_4^n$, if the boundary bases $x_1$ and $x_n$ are known, then $\bm{x}$ is uniquely determined by $L_{\mathcal{S}}(\bm{x})$. Moreover, among all label sets consisting of length-two labels, $\mathcal{S}$ is the minimum-cardinality set with this reconstruction property.
	\end{theorem}
	
	\begin{example}
		Let $\bm{x} = \mathrm{ACGA}$, and suppose both boundary bases ($x_1=\mathrm{A}$ and $x_4=\mathrm{A}$) are known. In $\mathcal{S}$, $\mathrm{AC}$ has index $1$, $\mathrm{CG}\notin\mathcal{S}$ contributes a $0$, and $\mathrm{GA}$ has index $3$. Hence $L_{\mathcal{S}}(\bm{x}) = (1,0,3,0)$.
		
		\textbf{Reconstruction.} Given $\bm{x} = \mathrm{A??A}$ and $L_{\mathcal{S}}(\bm{x}) = (1,0,3,0)$:
		\begin{enumerate}
			\item $c_1=1$ corresponds to $\mathrm{AC}$, so $\bm{x}_{[1:2]} = \mathrm{AC}$, i.e., $\bm{x} = \mathrm{AC?A}$.
			\item $c_3=3$ corresponds to $\mathrm{GA}$, so $\bm{x}_{[3:4]} = \mathrm{GA}$, i.e., $\bm{x} = \mathrm{ACGA}$.
		\end{enumerate}
		The labeling sequence couples consecutive bases, so the boundary bases propagate the reconstruction inward in a domino fashion.
	\end{example}
	
	\subsection{Error Models for Labeling Codes}
	
	We first recall the general indel error model of~\cite{2025ISIT}, in which errors act on the labeling sequence \emph{after} the labeling process.
	
	\begin{definition}[General Indel Error Model \cite{2025ISIT}]
		Let $\bm{e} = (e_1, e_2, e_3) \in \mathbb{N}^3$, and let $\mathcal{A}$ be a label set over $\Sigma_4$. For $\bm{x} \in \Sigma_4^n$, the \emph{$\bm{e}$-error $\mathcal{A}$-labeling ball} $\mathcal{B}L_{\mathcal{A}}(\bm{x}, \bm{e})$ is the set of all labeling sequences that can result from at most $e_1$ substitutions, $e_2$ insertions, and $e_3$ deletions to $L_{\mathcal{A}}(\bm{x})$.
		
		A code $\mathcal{C} \subseteq \Sigma_4^n$ is an \emph{$\bm{e}$-error $\mathcal{A}$-labeling code} if for any distinct $\bm{x}_1, \bm{x}_2 \in \mathcal{C}$,
		\begin{equation*}
			\mathcal{B}L_{\mathcal{A}}(\bm{x}_1, \bm{e}) \cap \mathcal{B}L_{\mathcal{A}}(\bm{x}_2, \bm{e}) = \emptyset.
		\end{equation*}
		Special cases give substitution, insertion, and deletion $\mathcal{A}$-labeling codes when the other two error counts vanish.
	\end{definition}
	
	Unlike the scattered-error, this paper considers a single \emph{burst} of $t$ consecutive deletions or insertions. 
	
	\begin{definition}[Burst $t$-Deletion/Insertion $\mathcal{A}$-Labeling Ball]
		\label{def:burst_ball}
		For $t \ge 1$ and $\bm{x} \in \Sigma_q^n$, the \emph{burst $t$-deletion $\mathcal{A}$-labeling ball} $\mathcal{B}L_{\mathcal{A}}^{\text{del}}(\bm{x}, t)$ (resp.\ \emph{burst $t$-insertion $\mathcal{A}$-labeling ball} $\mathcal{B}L_{\mathcal{A}}^{\text{ins}}(\bm{x}, t)$) is the set of sequences obtainable by deleting (resp.\ inserting) exactly $t$ consecutive symbols from $L_{\mathcal{A}}(\bm{x})$.
	\end{definition}
	
	\begin{definition}[Burst $t$-Deletion/Insertion $\mathcal{A}$-Labeling Code]
		\label{def:burst_code}
		Let $t \ge 1$ be an integer. A code $\mathcal{C} \subseteq \Sigma_q^n$ is a \emph{burst $t$-deletion/insertion $\mathcal{A}$-labeling code} if for any distinct $\bm{x}_1, \bm{x}_2 \in \mathcal{C}$, both
		\begin{equation*}
			\mathcal{B}L_{\mathcal{A}}^{\text{del}}(\bm{x}_1, t) \cap \mathcal{B}L_{\mathcal{A}}^{\text{del}}(\bm{x}_2, t) = \emptyset
		\end{equation*}
		and
		\begin{equation*}
			\mathcal{B}L_{\mathcal{A}}^{\text{ins}}(\bm{x}_1, t) \cap \mathcal{B}L_{\mathcal{A}}^{\text{ins}}(\bm{x}_2, t) = \emptyset
		\end{equation*}
		hold.
	\end{definition}
	
	The efficiency of such a code is measured by its \emph{redundancy}~\cite{1984VT_def}, defined below; it counts, in $\Sigma_q$-symbols, how much information the code sacrifices for error correction relative to the unconstrained space $\Sigma_q^n$.

	\begin{definition}[Redundancy~\cite{1984VT_def}]
		\label{def:redundancy}
		For a code $\mathcal{C}\subseteq\Sigma_q^n$, the \emph{redundancy} of $\mathcal{C}$ is
		\begin{equation*}
			r(\mathcal{C}) \triangleq n - \log_q |\mathcal{C}|.
		\end{equation*}
	\end{definition}
	
	\subsection{Classical VT and Shifted-VT Codes}
	\label{subsec:1d_insdel_codes}
	
	The Varshamov--Tenengolts (VT) code and its shifted variant (SVT) are the row-level building blocks of our construction.
	
	\begin{definition}[$q$-ary VT Codes \cite{1984VT_def}]
		For integers $n \ge 1$, $q \ge 2$, $0 \le a < n$, and $0 \le b < q$, define
		\begin{equation*}
    \mathrm{VT}_{a,b}(n, q) \triangleq \bigg\{ \bm{x} \in \Sigma_q^n : \sum_{i=1}^n (i-1)\alpha_i \equiv a \pmod n, \quad \sum_{i=1}^n x_i \equiv b \pmod q \bigg\},
\end{equation*}
	\end{definition}
	
	\begin{theorem}[\cite{1984VT_def}]
		$\mathrm{VT}_{a,b}(n, q)$ can correct a single insertion or deletion in $\mathcal{O}(n)$ time.
	\end{theorem}
	
	\begin{definition}[$q$-ary Shifted VT Codes \cite{2017QVT_def,2017q-ary-QVT_def}]
		For integers $n \ge 1, q \ge 2$, $0 \le a \le P$, $0 \le b < q$, and $c \in \{0, 1\}$, the $q$-ary shifted VT code $\mathrm{SVT}_{a,b,c}(n, P, q)$ is the set of $\bm{x}\in\Sigma_q^n$ satisfying
		\begin{equation*}
			\left\{
			\begin{aligned}
				\sum_{i=1}^n i \alpha_i &\equiv a \pmod{P + 1}, \\
				\sum_{i=1}^n x_i &\equiv b \pmod q, \\
				\sum_{i=1}^n \alpha_i &\equiv c \pmod 2,
			\end{aligned}
			\right.
		\end{equation*}
		where $\bm{\alpha}$ is the same characteristic sequence as for VT codes.
	\end{definition}

\begin{theorem}[\cite{2017q-ary-QVT_def}]
    The $q$-ary SVT code $\mathrm{SVT}_{a,b,c}(n, P, q)$ can correct a single insertion or deletion, provided the exact index of the error is known a priori to fall within a specific interval of $P$ consecutive positions, i.e., within $[w, w + P - 1]$ for some known integer $w$ ($1 \le w \le n$). The complexity of this decoding algorithm is $\mathcal{O}(n)$.
\end{theorem}
	
	SVT codes have substantially lower redundancy than VT codes at the cost of requiring a coarse a-priori localization of the error. In our matrix construction this is exploited as follows: a VT code in the first row localizes the burst, after which all remaining rows can be decoded by high-rate SVT codes within a small window, for details see Section~\ref{sec:dec}.

	\section{Fundamental Limit of Burst $t$-Deletion/Insertion $\mathcal{A}$-Labeling Codes} \label{sec:theoretical_bound}

    We now establish a sphere-packing-style lower bound on the redundancy of any burst $t$-deletion/insertion $\mathcal{A}$-labeling code. A direct application of standard $q$-ary burst arguments yields a valid but loose lower bound of $\log_{11} n + \mathcal{O}(1)$. To achieve a tighter bound of $\log_4 n + \mathcal{O}(1)$ that accurately reflects the fundamental limit, our proof must circumvent the non-surjectivity of $L_{\mathcal{A}}$ by operating strictly over its valid image space. We accomplish this by combining an interleaving decomposition with a delicate counting argument on long runs within the underlying DNA sequence. The formal statement and proof are as follows.

	
	\begin{theorem} \label{thm:theoretical_bound_new}
		For $t \ge 1$ and $t \mid n$, the redundancy of any burst $t$-deletion/insertion $\mathcal{A}$-labeling code $\mathcal{C}\subseteq\Sigma_4^n$ satisfies
		\begin{equation*}
			r(\mathcal{C}) \ge \log_4 n + \mathcal{O}(1).
		\end{equation*}
	\end{theorem}

    \begin{IEEEproof}
	By Definition~\ref{def:burst_code}, a burst $t$-deletion/insertion $\mathcal{A}$-labeling code must strictly satisfy the disjointness conditions for both error types. Consequently, the redundancy required to satisfy the deletion constraint alone inherently serves as a lower bound for the joint error model. Thus, it is sufficient to establish the fundamental limit by restricting our analysis exclusively to burst $t$-deletions.
	
	Let $\mathcal{C}\subseteq\Sigma_4^n$ be any burst $t$-deletion/insertion $\mathcal{A}$-labeling code. For $\bm{x}\in\mathcal{C}$, let $\bm{y}=L_{\mathcal{A}}(\bm{x})\in\Sigma_{11}^n$. A burst $t$-deletion reduces $\bm{y}$ to length $n-t$.
	
	\emph{Step 1: A reduced received space.}
	Let $L'_{\mathcal{A}}:\Sigma_4^{n-t}\to\Sigma_{11}^{n-t}$ denote the labeling map applied to length-$(n-t)$ DNA strands (using the same rule as $L_{\mathcal{A}}$), and let the \emph{valid received space} be its image:
	\begin{equation*}
		\mathcal{R}_{\text{valid}} \triangleq L'_{\mathcal{A}}\big(\Sigma_4^{n-t}\big).
	\end{equation*}
	Since $L'_{\mathcal{A}}$ is defined on $\Sigma_4^{n-t}$,
	\begin{equation}\label{eq:Rvalid_bound}
		|\mathcal{R}_{\text{valid}}| \le 4^{n-t}.
	\end{equation}
	Define the \emph{effective deletion ball} of $\bm{x}$ as
	\begin{equation*}
		\mathcal{B}_{\text{valid}}(\bm{x}) \triangleq \mathcal{B}L_{\mathcal{A}}^{\text{del}}(\bm{x}, t) \cap \mathcal{R}_{\text{valid}}.
	\end{equation*}
	By Definition~\ref{def:burst_code}, the balls $\mathcal{B}L_{\mathcal{A}}^{\text{del}}(\bm{x},t)$ are pairwise disjoint for distinct $\bm{x}\in\mathcal{C}$, hence so are their subsets $\mathcal{B}_{\text{valid}}(\bm{x})$. Combining with Eq.~\eqref{eq:Rvalid_bound}:
	\begin{equation} \label{eq:sphere_packing_valid}
		\sum_{\bm{x} \in \mathcal{C}} |\mathcal{B}_{\text{valid}}(\bm{x})| \le |\mathcal{R}_{\text{valid}}| \le 4^{n-t}.
	\end{equation}
	
	\emph{Step 2: A run-based lower bound on the effective ball.}
	Let $w \triangleq n/t$, and consider the first interleaved subsequence
	\begin{equation*}
		\bm{x}^{(1)} \triangleq (x_1, x_{1+t}, x_{1+2t}, \dots, x_{1+(w-1)t}) \in \Sigma_4^w.
	\end{equation*}
	Denote by $R_{\ge 2}(\bm{x}^{(1)})$ the number of runs of length at least $2$ in $\bm{x}^{(1)}$. We claim that
	\begin{equation}\label{eq:ball_lower}
		|\mathcal{B}_{\text{valid}}(\bm{x})| \ge R_{\ge 2}(\bm{x}^{(1)}) - 2.
	\end{equation}

	To establish Eq.~\eqref{eq:ball_lower}, we construct an injection from a structurally constrained set of runs into $\mathcal{B}_{\text{valid}}(\bm{x})$. 
	
	Let $\mathcal{R}$ be the collection of runs of length $\ge 2$ in $\bm{x}^{(1)}$. To ensure the boundary bases $x_1$ and $x_n$ are strictly preserved during deletion, we define $\mathcal{R}^* \subseteq \mathcal{R}$ as the subset obtained by excluding any run that contains the first index ($1$) or the last index ($w$) of $\bm{x}^{(1)}$. Because $\mathcal{R}$ consists of maximal blocks, at most two runs touch these extremities, yielding $|\mathcal{R}^*| \ge R_{\ge 2}(\bm{x}^{(1)}) - 2$.

	\emph{(i) Each internal run yields a valid deletion index.}
	Each run $\mathsf{r}\in\mathcal{R}^*$ starts at some index $s$ in $\bm{x}^{(1)}$. By our boundary exclusion, $2 \le s \le w-1$. The corresponding start index in the original DNA strand $\bm{x}$ is $k_{\mathsf{r}} \triangleq 1+(s-1)t$. The first two symbols of the run guarantee $x_{k_{\mathsf{r}}} = x_{k_{\mathsf{r}}+t}$. 
	
	Let $\bm{x}_{\text{del}}(k_{\mathsf{r}})$ and $\bm{y}_{\text{del}}(k_{\mathsf{r}})$ be the sequences obtained by deleting $t$ consecutive symbols starting at index $k_{\mathsf{r}}$ from $\bm{x}$ and $\bm{y}$, respectively. Because $x_{k_{\mathsf{r}}} = x_{k_{\mathsf{r}}+t}$, every dinucleotide defining a label in $\bm{x}_{\text{del}}(k_{\mathsf{r}})$ retains its original value, so $L'_{\mathcal{A}}\big(\bm{x}_{\text{del}}(k_{\mathsf{r}})\big) = \bm{y}_{\text{del}}(k_{\mathsf{r}})$. Thus, $\bm{y}_{\text{del}}(k_{\mathsf{r}}) \in \mathcal{B}_{\text{valid}}(\bm{x})$.

	\emph{(ii) Distinct runs yield distinct deleted DNA strands.}
	Let $\mathsf{r},\mathsf{r}'\in\mathcal{R}^*$ be distinct, and write $k\triangleq k_{\mathsf{r}}<k_{\mathsf{r}'}\triangleq k'$ (so $k'-k\ge t$). Because $\mathsf{r}$ and $\mathsf{r}'$ are distinct runs, there exists an integer $m \ge 2$ where the value changes:
	\begin{equation*}
		x_{1+(m-1)t} \ne x_{1+mt}, \qquad k < 1+(m-1)t < k'.
	\end{equation*}
	Comparing the two deletions:
	\begin{itemize}
		\item In $\bm{x}_{\text{del}}(k)$, the position $1+(m-1)t$ is shifted left by $t$ to index $1+(m-2)t$.
		\item In $\bm{x}_{\text{del}}(k')$, the position $1+(m-1)t$ lies to the left of the deleted block and remains at index $1+(m-1)t$.
	\end{itemize}
	Consequently, the coordinate $1+(m-1)t$ of $\bm{x}_{\text{del}}(k)$ equals $x_{1+mt}$, whereas the same coordinate of $\bm{x}_{\text{del}}(k')$ equals $x_{1+(m-1)t}\ne x_{1+mt}$. Hence $\bm{x}_{\text{del}}(k)\ne\bm{x}_{\text{del}}(k')$.

	\emph{(iii) Injectivity transfers to the label domain.}
	For any $\mathsf{r} \in \mathcal{R}^*$, the deleted block spans indices $[k_{\mathsf{r}}, k_{\mathsf{r}}+t-1]$. Since $s \ge 2$, the smallest deleted index is $k_{\mathsf{r}} \ge 1+t \ge 2$, leaving $x_1$ untouched. Since $s \le w-1$, the largest deleted index is $k_{\mathsf{r}}+t-1 \le (1+(w-2)t)+t-1 = n-t$, leaving $x_n$ untouched.
	
	Therefore, both $\bm{x}_{\text{del}}(k)$ and $\bm{x}_{\text{del}}(k')$ are length-$(n-t)$ DNA strands sharing the same undisturbed boundary bases $x_1$ and $x_n$. By Theorem~\ref{thm:reconstruction}, $L'_{\mathcal{A}}$ is injective on length-$(n-t)$ strands with fixed boundary bases, so
	\begin{equation*}
		\bm{y}_{\text{del}}(k) = L'_{\mathcal{A}}(\bm{x}_{\text{del}}(k)) \;\ne\; L'_{\mathcal{A}}(\bm{x}_{\text{del}}(k')) = \bm{y}_{\text{del}}(k').
	\end{equation*}

	Combining (i)--(iii), the map $\mathsf{r}\mapsto\bm{y}_{\text{del}}(k_{\mathsf{r}})$ is an injection from $\mathcal{R}^*$ into $\mathcal{B}_{\text{valid}}(\bm{x})$, proving Eq.~\eqref{eq:ball_lower}.
	
	\emph{Step 3: Few codewords have $R_{\ge 2}(\bm{x}^{(1)}) \le \rho$.}
	Set $\delta = 0.03$ and $\rho \triangleq \lfloor \delta w \rfloor$, and partition $\mathcal{C}$ into
	\begin{align*}
		\mathcal{C}_{\text{low}} &\triangleq \{\bm{x} \in \mathcal{C} : R_{\ge 2}(\bm{x}^{(1)}) \le \rho\}, \\
		\mathcal{C}_{\text{high}} &\triangleq \mathcal{C}\setminus\mathcal{C}_{\text{low}}.
	\end{align*}
	
	To bound $|\mathcal{C}_{\text{low}}|$, we first count sequences $\bm{u}\in\Sigma_4^w$ with $R_{\ge 2}(\bm{u})\le\rho$. Suppose such a $\bm{u}$ has exactly $K$ runs, of which exactly $j$ have length $\ge 2$, with $0\le j\le\min(\rho,K)$ and $1\le K \le w-j$ (since the total length is $w \ge (K-j) + 2j = K+j$). Then:
	\begin{itemize}
		\item the chain of $K$ run values has $4\cdot 3^{K-1}$ choices (adjacent runs carry distinct symbols);
		\item the $\binom{K}{j}$ ways to choose which runs are length $\ge 2$;
		\item the $\binom{w-K-1}{j-1}$ ways to distribute the $w-K$ extra units among the $j$ runs of length $\ge 2$ (which is equivalent to the number of positive integer solutions to the equation $x_1 + \dots + x_j = w-K$; defined as $0$ when $j=0$ and $w>K$).
	\end{itemize}
	Let $V_{\text{low}} \triangleq \big|\{\bm{u}\in\Sigma_4^w : R_{\ge 2}(\bm{u})\le\rho\}\big|$. Then
	\begin{equation*}
		V_{\text{low}} = \sum_{K=1}^{w} \sum_{j=0}^{\rho} 4 \cdot 3^{K-1} \binom{K}{j} \binom{w-K-1}{j-1}.
	\end{equation*}

	Setting $\delta = 0.03$ guarantees that $\rho \le 0.03w < w/2$. We bound the binomial coefficient $\binom{K}{j}$ under two conditions:
    \begin{enumerate}
		\item If $K \ge 2\rho$: Since $j \le \rho \le K/2$, due to the monotonicity of binomial coefficients in the first half of their domain, we have $\binom{K}{j} \le \binom{K}{\rho}$. Furthermore, since $K \le w$, it follows that $\binom{K}{\rho} \le \binom{w}{\rho}$.
		\item If $K < 2\rho$: Bounding by the central binomial coefficient yields $\binom{K}{j} \le \binom{K}{\lfloor K/2 \rfloor} \le \binom{2\rho}{\rho}$. Combined with the condition $2\rho \le w$, it holds that $\binom{2\rho}{\rho} \le \binom{w}{\rho}$.
	\end{enumerate}
	Therefore, $\binom{K}{j} \le \binom{w}{\rho}$ holds for all cases. By analogous reasoning, $\binom{w-K-1}{j-1} \le \binom{w}{\rho}$. Combined with $\sum_{K=1}^w 3^{K-1} < 3^w/2$, this gives
	\begin{equation*}
		V_{\text{low}} \le 4(\rho+1)\binom{w}{\rho}^2 \cdot \tfrac{3^w}{2} \le 4w\binom{w}{\rho}^2 \cdot 3^w.
	\end{equation*}
	
	By the binary-entropy bound~\cite[Theorem~11.1.3]{Cover2005}, $\binom{w}{\rho} \le 2^{w H(\delta)}$ with $H(\delta) = -\delta\log_2\delta-(1-\delta)\log_2(1-\delta)$. For $\delta=0.03$, $H(0.03)<0.195$, so
	\begin{equation*}
		V_{\text{low}} \le 4w \cdot \big(3\cdot 4^{H(0.03)}\big)^w \le 4w\cdot(3.93)^w.
	\end{equation*}
	Since $3.93 < 4$, $V_{\text{low}} = o(4^w/n^2)$. The other $n-w$ bases of $\bm{x}$ are unconstrained, so
	\begin{equation*}
		|\mathcal{C}_{\text{low}}| \le 4^{n-w}\cdot V_{\text{low}} = o\!\left(\frac{4^n}{n^2}\right).
	\end{equation*}
	
	\emph{Step 4: Sphere-packing for $\mathcal{C}_{\text{high}}$.}
	For every $\bm{x}\in\mathcal{C}_{\text{high}}$, $R_{\ge 2}(\bm{x}^{(1)}) > \rho \ge \delta w - 1$, so by Eq.~\eqref{eq:ball_lower},
	\begin{equation*}
		|\mathcal{B}_{\text{valid}}(\bm{x})| \ge \rho - 2 \ge \delta w - 3 \ge \frac{0.03 n}{t} - 3.
	\end{equation*}
	Substituting into Eq.~\eqref{eq:sphere_packing_valid}:
	\begin{equation*}
		|\mathcal{C}_{\text{high}}|\cdot\bigg(\frac{0.03 n}{t} - 3\bigg)\le\sum_{\bm{x} \in \mathcal{C}_{\text{high}}} |\mathcal{B}_{\text{valid}}(\bm{x})|  \le 4^{n-t},
	\end{equation*}
	hence $|\mathcal{C}_{\text{high}}| = \mathcal{O}(4^{n-t}/n)$.
	
	\emph{Step 5: Conclusion.}
	Combining Steps~3 and~4,
	\begin{equation*}
		|\mathcal{C}| = |\mathcal{C}_{\text{low}}| + |\mathcal{C}_{\text{high}}| = \mathcal{O}\!\left(\frac{4^n}{n}\right),
	\end{equation*}
	so $r(\mathcal{C}) = n - \log_4|\mathcal{C}| \ge \log_4 n + \mathcal{O}(1)$, completing the proof.
\end{IEEEproof}

    It is worth noting that while the sphere-packing bound in Theorem 4 is explicitly derived for the minimal length-two label set $\mathcal{S}$ (i.e., $l=2$), the fundamental lower bound of $\log_4 n + \mathcal{O}(1)$ remains robust against variations in the label length. Specifically, the combinatorial framework can be generalized to any fixed label length $l \ge 2$, entirely independent of the burst length $t$. For completeness, a rigorous and step-by-step proof of this extension is provided in Appendix \ref{sec:appendix_generalization}.
	
	\section{Generalized RLL Codes}
	\label{sec:generalized_rll}
	
	In this section, we introduce \emph{generalized RLL codes}, a DNA-domain constraint that deterministically enforces an ordinary run-length-limited (RLL) property in the resulting label sequence. This indirection is necessary because the encoder may only design $\bm{x}\in\Sigma_4^n$, not the labeling output $L_{\mathcal{S}}(\bm{x})\in\Sigma_{11}^n$.
	
	\subsection{Design Motivation and DNA Pair Mapping}
	
	A deletion or insertion within a long run of identical symbols is positionally ambiguous, which hurts row-level error localization in our matrix construction. To bound this ambiguity, certain label-domain subsequences must satisfy the classical $\ell$-RLL constraint: every run has length at most $\ell$.
	
	Since the encoder controls only $\bm{x}$, we need a constraint on $\bm{x}$ that \emph{implies} an $\ell$-RLL property on $L_{\mathcal{S}}(\bm{x})$. We achieve this through the \emph{generalized $\ell$-RLL constraint}, whose implication relation will be proved in Theorem~\ref{thm:gen_rll_implication}.
	
	Because the label at position $i$ is determined by the DNA dinucleotide $(x_i,x_{i+1})$, we group $\bm{x}$ into adjacent base pairs and map them to a $16$-ary alphabet, which becomes the natural alphabet on which our constraint will be stated.
	
	\begin{definition}[DNA Pair Mapping]\label{def:DNApair}
		A \emph{DNA pair} is an ordered pair $(a,b)\in\Sigma_4\times\Sigma_4$. The bijection
		\begin{equation*}
			\psi:\Sigma_4\times\Sigma_4\to\Sigma_{16},\qquad \psi(a,b) = 4a+b,
		\end{equation*}
		represents any even-length DNA sequence $\bm{x}=(x_1,\ldots,x_{2k})$ as a length-$k$ sequence $\bm{z}=(z_1,\ldots,z_k)\in\Sigma_{16}^k$ with $z_j = \psi(x_{2j-1},x_{2j})$.
	\end{definition}
	
	\subsection{Definition of the Generalized RLL Constraint}
	\label{subsec:gen_RLL_def_property}
	
	To carry the RLL idea over to sequences in $\Sigma_{16}$, we group $\Sigma_{16}$ according to the label each pair generates. Define $\phi_{\text{label}}:\Sigma_{16}\to\Sigma_{11}$ by
	\begin{equation*}
		\phi_{\text{label}}(z) \triangleq \begin{cases}
			k, & \text{if } \psi^{-1}(z) \text{ is the } k\text{-th label in } \mathcal{S}, \\
			0, & \text{if } \psi^{-1}(z) \notin \mathcal{S},
		\end{cases}
	\end{equation*}
	divide $\Sigma_{16}$ into equivalence classes of $\phi_{\text{label}}$. This yields exactly $11$ classes:
	
	\begin{enumerate}
		\item \textbf{Special class $\mathbb{C}_{spec}$.}
		The six DNA pairs $\{\mathrm{AA, AT, AG, CA, CC, CG}\}$ all map to label $0$, so
		\begin{equation*}
			\mathbb{C}_{spec} = \{0, 1, 3, 9, 10, 11\} \subset \Sigma_{16}.
		\end{equation*}
		\item \textbf{Singleton normal classes $\{\mathbb{C}_v\}_{v\in\mathcal{N}}$.}
		Let $\mathcal{N} \triangleq \Sigma_{16}\setminus\mathbb{C}_{spec} = \{2,4,5,6,7,8,12,13,14,15\}$. Each $v\in\mathcal{N}$ maps to a distinct non-zero label, so each forms its own class $\mathbb{C}_v = \{v\}$.
	\end{enumerate}

    Based on this classification, we define the generalized RLL
constraint. Unlike traditional RLL, which limits consecutive
identical symbols, generalized RLL limits consecutive symbols
from the same equivalence class.
	
	\begin{definition}[Generalized $\ell$-RLL Constraint]
		A sequence $\bm{z}\in\Sigma_{16}^n$ satisfies the \emph{generalized $\ell$-RLL constraint} if no substring of $\bm{z}$ of length $\ell+1$ is contained in a single equivalence class.
	\end{definition}
	
	In other words, no class may contribute more than $\ell$ consecutive symbols.
	
	\subsection{Properties of Generalized RLL Codes}
	
	The primary goal of imposing the generalized RLL constraint on the $\Sigma_{16}$ sequence is to control the properties of the resulting labeling sequence. The following theorem establishes this relationship.
	
	\begin{theorem}[Implication of Generalized RLL]
		\label{thm:gen_rll_implication}
		Let $\bm{z} \in \Sigma_{16}^n$ be a sequence satisfying the generalized $\ell$-RLL constraint. Let $\bm{y} \in \Sigma_{11}^n$ be the corresponding labeling sequence defined by $y_i = \phi_{\text{label}}(z_i)$. Then, $\bm{y}$ satisfies the traditional $\ell$-RLL constraint.
	\end{theorem}
	
	\begin{IEEEproof}
		Suppose for the sake of contradiction that $\bm{y}$ violates the traditional $\ell$-RLL constraint. This signifies the existence of a substring $y_i, \dots, y_{i+\ell}$ of length $\ell+1$, where all elements are identical to a specific value $v \in \Sigma_{11}$.
		
		According to the definition of equivalence classes, this implies that the corresponding source symbols $z_i, \dots, z_{i+\ell}$ are all mapped to the exact same label value $v$. Consequently, they must all belong to the same equivalence class (specifically, $\mathbb{C}_{spec}$ if $v=0$, or a particular $\mathbb{C}_v$ if $v>0$). 
		
		This forms a run of length $\ell+1$ comprising symbols from the identical class, directly violating the premise that $\bm{z}$ inherently satisfies the generalized $\ell$-RLL constraint. Thus, we reach a contradiction, completing the proof.
	\end{IEEEproof}
	
	\subsection{Coding Capacity Analysis}\label{subsec:gen_RLL_capacity}
	
	To bound the capacity of generalized RLL codes, we count the number of valid sequences of length $n$ using two state variables:
	
	\begin{itemize}
		\item $S[n, \ell]$: number of generalized $\ell$-RLL sequences of length $n$ ending in some symbol of $\mathbb{C}_{spec}$.
		\item $N[n, \ell]$: number of generalized $\ell$-RLL sequences of length $n$ ending in some normal class $\mathbb{C}_v$.
	\end{itemize}
	
	By symmetry, the number of valid sequences ending in a fixed $\mathbb{C}_v$ is $N[n,\ell]/10$. These counts obey the following recurrence.
	
	\begin{lemma}[Recurrence relations]
		\label{lem:recur}
		For $n \ge 1$,
		\begin{align}
			S[n, \ell] &= \sum_{i=1}^{\min(n, \ell)} 6^i \cdot N[n-i, \ell], \label{eq:rec_S} \\
			N[n, \ell] &= \sum_{i=1}^{\min(n, \ell)} 10 \cdot \left( S[n-i, \ell] + \tfrac{9}{10}N[n-i, \ell] \right), \label{eq:rec_N}
		\end{align}
		with formal initial values $N[0,\ell] = 1$ and $S[0,\ell] = 1/10$. The fractional value $S[0,\ell]=1/10$ is not a count; it is an algebraic convenience that makes Eqs.~\eqref{eq:rec_S}--\eqref{eq:rec_N} hold uniformly for $n\ge 1$, and produces the same $S[n,\ell],N[n,\ell]$ values as the equivalent integer base $N[1,\ell]=10,\,S[1,\ell]=6$.
	\end{lemma}
	
	\begin{IEEEproof}
		\emph{Derivation of Eq.~\eqref{eq:rec_S}.}
		Decompose each valid sequence counted by $S[n,\ell]$ according to the length $i$ of its maximal trailing $\mathbb{C}_{spec}$-run; the constraint forces $1\le i\le\min(n,\ell)$. The trailing $i$ symbols are arbitrary in $\mathbb{C}_{spec}$, giving $6^i$ choices. The preceding symbol (if any) must lie in a normal class, so its prefix has $N[n-i,\ell]$ choices. Summing over $i$ yields Eq.~\eqref{eq:rec_S}.
		
		\emph{Derivation of Eq.~\eqref{eq:rec_N}.}
		Fix a specific normal class $\mathbb{C}_v$ and let $i$ be the length of the trailing $\mathbb{C}_v$-run. Since $|\mathbb{C}_v|=1$, this segment is uniquely determined. The preceding symbol must lie either in $\mathbb{C}_{spec}$ (contributing $S[n-i,\ell]$ prefixes) or in a normal class other than $\mathbb{C}_v$ (contributing $\frac{9}{10}N[n-i,\ell]$ prefixes by symmetry). Summing over $i$ and multiplying by the $10$ choices of $v\in\mathcal{N}$ yields Eq.~\eqref{eq:rec_N}.
	\end{IEEEproof}
	
	Let $\mathcal{C}(n, \ell)$ represent the set of all valid sequences of length $n$ satisfying the generalized $\ell$-RLL constraint. The total cardinality of this set, denoted as $|\mathcal{C}(n,\ell)|$, is therefore the sum of the sequences of the two terminal states:
	\begin{equation*}
		|\mathcal{C}(n, \ell)| = S[n, \ell] + N[n, \ell].
	\end{equation*}
	
	To evaluate the impact of the proposed generalized RLL constraint on encoding efficiency, we investigate its asymptotic code rate. The code rate is formally defined as:
	\begin{equation*}
		r_{\mathcal{C}(n, \ell)} \triangleq \frac{\log_{16}(|\mathcal{C}(n, \ell)|)}{n}.
	\end{equation*}
	The base-16 logarithm is used because each sequence symbol represents a DNA base pair (an element of $\Sigma_{16}$), which corresponds to $4^2=16$ possible combinations.
	
	The following theorem establishes the exact asymptotic code rate as the sequence length approaches infinity.
	
	\begin{theorem}\label{thm:genRLL_redundancy}
		For a given maximum run length $\ell$, the asymptotic code rate of the generalized $\ell$-RLL constraint is given by:
		\begin{equation*}
			\lim_{n \to \infty} r_{\mathcal{C}(n, \ell)} = \log_{16}\lambda,
		\end{equation*}
		where $\lambda$ is the largest positive real root of the following characteristic equation:
		\begin{equation*}
			10 A\left(\frac{1}{\lambda}\right) B\left(\frac{1}{\lambda}\right) + 9 B\left(\frac{1}{\lambda}\right) = 1.
		\end{equation*}
		Here, the auxiliary polynomials are defined as $A(x) = \sum_{i=1}^\ell 6^i x^i$ and $B(x) = \sum_{i=1}^\ell x^i$.
	\end{theorem}
	
	\begin{IEEEproof}
		For brevity in the subsequent proof, we denote the state counts $S[n,\ell]$ and $N[n,\ell]$ simply as $S(n)$ and $N(n)$, respectively. We define the generating functions $G_S(x) = \sum_{n=0}^\infty S(n)x^n$ and $G_N(x) = \sum_{n=0}^\infty N(n)x^n$.
		
		First, we convert the recurrence relation Eq.~\eqref{eq:rec_S} into its generating function form:
		\begin{align*}
			A(x)G_N(x) &= \left(\sum_{i=1}^\ell 6^i x^i\right) \left(\sum_{n=0}^\infty N(n)x^n\right) \nonumber \\
			&= \sum_{n=0}^\infty \sum_{i=1}^\ell 6^i N(n) x^{n+i} \nonumber \\
			&\overset{m=n+i}{=} \sum_{m=1}^\infty \left( \sum_{i=1}^{\min(m, \ell)} 6^i N(m-i) \right) x^m \nonumber \\
			&\overset{\text{Eq.~\eqref{eq:rec_S}}}{=} \sum_{m=1}^\infty S(m) x^m \nonumber \\
			&= G_S(x) - S(0).
		\end{align*}
		Using the initial condition $S(0) = \frac{1}{10}$, we obtain:
		\begin{equation} \label{eq:gen_S}
			G_S(x) = A(x) G_N(x) + \frac{1}{10}.
		\end{equation}
		
		Similarly, for the recurrence relation Eq.~\eqref{eq:rec_N}, we have:
		\begin{align*}
			B(x)G_S(x) &= \sum_{m=1}^\infty \left( \sum_{i=1}^{\min(m, \ell)} S(m-i) \right) x^m, \\
			B(x)G_N(x) &= \sum_{m=1}^\infty \left( \sum_{i=1}^{\min(m, \ell)} N(m-i) \right) x^m.
		\end{align*}
		Combining these two equations and multiplying by their respective coefficients yields:
		\begin{align*}
			& 10 B(x) G_S(x) + 9 B(x) G_N(x) \nonumber \\
			&= \sum_{m=1}^\infty \sum_{i=1}^{\min(m, \ell)} \left[ 10 S(m-i) + 9 N(m-i) \right] x^m \nonumber \\
			&= \sum_{m=1}^\infty N(m) x^m \quad (\text{by Eq.~\eqref{eq:rec_N}}) \nonumber \\
			&= G_N(x) - N(0).
		\end{align*}
		Substituting the initial condition $N(0)=1$ and rearranging, we get:
		\begin{equation} \label{eq:gen_N_temp}
			10 B(x) G_S(x) + 9 B(x) G_N(x) = G_N(x) - 1.
		\end{equation}
		
		Substituting Eq.~\eqref{eq:gen_S} into Eq.~\eqref{eq:gen_N_temp} to eliminate $G_S(x)$ gives:
		\begin{equation*}
			10 B(x) \left[ A(x) G_N(x) + \frac{1}{10} \right] + 9 B(x) G_N(x) = G_N(x) - 1.
		\end{equation*}
		Rearranging this equation to solve for $G_N(x)$, we obtain:
		\begin{align*}
			G_N(x) = \frac{ B(x) + 1}{1 - B(x) (9 + 10 A(x))}.
		\end{align*}
		Subsequently, we can express $G_S(x)$ as:
		\begin{equation*}
			G_S(x) = \frac{A(x)B(x)+A(x)}{1 - 9 B(x) - 10 A(x) B(x)}+\frac{1}{10}.
		\end{equation*}
		Defining the generating function for the total number of sequences as $G_{\mathcal{C}}(x) \triangleq \sum_{n=0}^\infty |\mathcal{C}(n,\ell)|x^n$, we have:
		\begin{equation*}
			G_{\mathcal{C}}(x) \!=\! G_S(x) + G_N(x) \!=\! \frac{(1 + A(x))(1+B(x))}{1 - 9 B(x) - 10 A(x) B(x)}+\frac{1}{10}.
		\end{equation*}
		
		The coefficient of the $x^n$ term in $G_{\mathcal{C}}(x)$ represents the desired value $|\mathcal{C}(n,\ell)|$. To determine the asymptotic growth rate $\lambda$ of $|\mathcal{C}(n,\ell)|$ (i.e., $|\mathcal{C}(n,\ell)| \approx C \lambda^n$), we apply the Rational Function Coefficient Asymptotics theorem \cite[Theorem IV.9]{Algebra2009}:
		If $G(x) = \frac{P(x)}{Q(x)}$ is a rational function analytic at zero, and $\alpha$ is the root of $Q(x)$ with the smallest modulus (the dominant pole), then $[x^n]G(x) \sim C \cdot (1/\alpha)^n$, where $[x^n]G(x)$ denotes the coefficient of $x^n$ in $G(x)$.
		
		Therefore, the total number of sequences satisfies $|\mathcal{C}(n, \ell)|\sim C \cdot (1/\alpha)^n$, where $\alpha$ is the smallest positive root of the denominator $1-9 B(x) -10 A(x) B(x)$. Consequently, $r_{\mathcal{C}(n, \ell)} \sim \log_{16}\frac{1}{\alpha}$.
		By letting $\lambda = 1/\alpha$, we obtain $r_{\mathcal{C}(n, \ell)} \sim \log_{16}\lambda$, where $\lambda$ is the largest positive real root of the characteristic equation:
		\begin{equation*}
			1 = 9 B\left(\frac{1}{\lambda}\right) + 10 A\left(\frac{1}{\lambda}\right) B\left(\frac{1}{\lambda}\right).
		\end{equation*}
		This completes the proof.
	\end{IEEEproof}

	\section{Encoding and Decoding Scheme for Generalized RLL Codes}
	\label{sec:grll_coding_scheme}
	
	This section details the encoding and decoding procedures for the generalized RLL codes. We employ the enumerative coding technique \cite{2021RLL}, whose fundamental principle is to establish a strict bijection between the information space (a set of integer sequences) and the codeword space (the set of sequences satisfying the constraint).
	
	To realize this bijection, we design two core auxiliary algorithms: \emph{ranking} and \emph{unranking}. It is important to note that these algorithms serve solely as mathematical subroutines invoked by the encoder and decoder. Their relationship with the overall process is as follows:
	
	\begin{itemize}
		\item \textbf{Encoding Algorithm}: Converts the input information bitstream into an integer index, and subsequently calls the unranking algorithm to generate the corresponding valid codeword.
		\item \textbf{Decoding Algorithm}: Receives the codeword, calls the ranking algorithm to compute its lexicographical rank (integer index), and thereby recovers the original information bits.
	\end{itemize}

    To understand these two core subroutines intuitively, imagine that all valid sequences satisfying the generalized $\ell$-RLL constraint are sorted lexicographically in a dictionary. The \emph{ranking} algorithm acts as a lookup function: given a valid sequence, it calculates its exact integer index (its ``rank'') within this dictionary. Conversely, the \emph{unranking} algorithm performs the exact reverse: given an index, it reconstructs the specific valid sequence located at that position. Furthermore, although encoding practically precedes decoding in a communication pipeline, we deliberately present the ranking algorithm first in the following subsections. The reason for this ordering is that ranking straightforwardly introduces the lexicographical traversal logic and demonstrates how precomputed tables are used to count bypassed branches. Once the intuition of accumulating capacities to find a sequence's position is clear, the unranking process---which dynamically uses these identical branch capacities to construct a sequence from a target index---becomes a natural and intuitive extension.
	
	\subsection{Definitions and Notation}
	
	To eliminate ambiguity and unify our presentation, we define the following notation and auxiliary functions:
	
	\begin{itemize}
		\item \textbf{Sequence Decomposition}: 
		For a sequence $\bm{x}= (x_1, x_2, \ldots, x_n)\in \Sigma_{16}^n$, we define its \textit{suffix} as $\bm{x}_{suf} = (x_i, \ldots, x_n)$, where $i$ is the smallest index such that $x_i, \ldots, x_n$ all belong to the same equivalence class (i.e., they all belong to $\mathbb{C}_{spec}$ or a specific $\mathbb{C}_v$). The corresponding \textit{prefix} is defined as $\bm{x}_{pre} = (x_1, \ldots, x_{i-1})$.
		
		\item \textbf{Forbidden Constraint ($F$)}:
		We define the variable $F$ as the ``forbidden continuation type'' of the prefix. Its value domain is $\{\mathbb{C}_{spec}\} \cup \{\mathbb{C}_v \mid v \in \mathcal{N}\} \cup \{\emptyset\}$. This variable is utilized during the recursive process to ensure that the concatenated sequence strictly complies with the generalized RLL constraint. Specifically, $F = \emptyset$ signifies that no equivalence class is currently forbidden.
		
		\item \textbf{Special Class Mappings}:
		\begin{itemize}
			\item $\text{EncodeSpecial}(\bm{s})$: Maps a special class sequence $\bm{s}$ of length $len$ (composed entirely of elements from $\mathbb{C}_{spec}$) to an integer $val \in [0, 6^{len}-1]$. In practice, this is achieved by mapping the 6 elements $\{0, 1, 3, 9, 10, 11\}$ in $\mathbb{C}_{spec}$ to $\{0, \dots, 5\}$ in ascending order and then evaluating $\bm{s}$ as a base-6 integer.
			\item $\text{DecodeSpecial}(val, len)$: The inverse mapping of $\text{EncodeSpecial}$, which recovers the special class sequence $\bm{s}$ of length $len$ given the integer $val$.
		\end{itemize}
		
	\end{itemize}

    \begin{example}
        Consider a sequence $\bm{s} = (3, 10) \in \mathbb{C}_{spec}^2$ of length $len = 2$. By mapping the elements of $\mathbb{C}_{spec}$ to $\{0, \dots, 5\}$ in ascending order, the values $3$ and $10$ correspond to $2$ and $4$, respectively. To encode, we evaluate this mapped sequence as a base-6 integer: $\text{EncodeSpecial}(\bm{s}) = 2 \cdot 6^1 + 4 \cdot 6^0 = 16$. Conversely, for $\text{DecodeSpecial}(16, 2)$, we take the integer $16$, convert it back into a length-2 base-6 sequence $(2, 4)_6$, and map the digits back to their original values in $\mathbb{C}_{spec}$ to perfectly recover $(3, 10)$.
    \end{example}
	
	\subsection{Ranking Algorithm (for Decoding)}
	
	Before presenting the formal algorithmic steps, we briefly introduce the core intuition behind the ranking algorithm. Conceptually, enumerative coding organizes all valid generalized RLL sequences into a lexicographical tree. The rank of a sequence corresponds to its absolute index within this ordered tree.
	
	To compute this rank, the algorithm iteratively processes the target sequence by dividing it into prefixes and suffixes. At each step, it calculates the number of valid sequences in the branches that lexicographically precede the current suffix. By summing the total capacity of these bypassed branches—which are efficiently retrieved in $\mathcal{O}(1)$ time from the precomputed tables $S[n,\ell]$ and $N[n,\ell]$—the algorithm determines the exact global position of the sequence.
	
	The lexicographical traversal order for these branches is defined as follows:
	\begin{enumerate}
		\item Iterate through the suffix length $i$ (from $1$ to $\ell$);
		\item Traverse the suffix types, prioritizing the special class $\mathbb{C}_{spec}$, followed by the normal classes $\mathbb{C}_v$ in numerical order.
	\end{enumerate}
	
	Based on this traversal logic, the formal procedure for evaluating the rank of a sequence (and its corresponding inverse unranking process) is detailed in Algorithm \ref{alg:ranking}.
	
	\begin{algorithm}[htbp]
		\caption{Ranking Algorithm: $Rank(n, \ell, \bm{x}, F)$}
		\label{alg:ranking}
		\begin{algorithmic}[1]
			\State \textbf{Input:} Sequence length $n$, run-length limit $\ell$, sequence to be ranked $\bm{x}$, forbidden type $F$
			\State \textbf{Output:} Rank $M \in \mathbb{Z}^+$
			
			\If{$n=0$} \Return 1 \EndIf
			
			\State Parse the suffix $\bm{x}_{suf}$ of $\bm{x}$ to obtain its length $len$ and type $type$
			\State $accum \leftarrow 0$
			
			\For{$i \leftarrow 1$ to $\min(n, \ell)$}
			\State \textit{// Branch 1: Special Class ($\mathbb{C}_{spec}$)}
			\If{$F \neq \mathbb{C}_{spec}$}
			\If{$i = len$ \textbf{and} $type = \mathbb{C}_{spec}$}
			\State $R_{pre} \leftarrow Rank(n-i, \ell, \bm{x}_{pre}, \mathbb{C}_{spec})$
			\State $V_{suf} \leftarrow \text{EncodeSpecial}(\bm{x}_{suf})$
			\State \Return $accum + (R_{pre} - 1) \times 6^i + V_{suf} + 1$
			\EndIf
			\State $count \leftarrow 6^i \times N[n-i, \ell]$
			\State $accum \leftarrow accum + count$
			\EndIf
			
			\State \textit{// Branch 2: Normal Classes (Traverse all $v \in \mathcal{N}$)}
			\For{$v \in \mathcal{N}$}
			\If{$F \neq \mathbb{C}_v$}
			\If{$i = len$ \textbf{and} $type = \mathbb{C}_v$}
			\State $R_{pre} \leftarrow Rank(n-i, \ell, \bm{x}_{pre}, \mathbb{C}_v)$
			\State \Return $accum + R_{pre}$
			\EndIf
			\State $count \leftarrow S[n-i, \ell] + \frac{9}{10} N[n-i, \ell]$
			\State $accum \leftarrow accum + count$
			\EndIf
			\EndFor
			\EndFor
		\end{algorithmic}
	\end{algorithm}
	
	To further clarify the execution flow and the recursive accumulation of the rank, we provide a step-by-step numerical illustration below.
	
	\begin{example}
		Assume $n=4, \ell=2$. Objective: Calculate the Rank of the sequence $\bm{x} = [4, 9, 4, 11]$, i.e., $Rank(4, 2, [4, 9, 4, 11], \emptyset)$.
		Assume the following precomputed counting tables (calculated via the recurrence relations):
		$S[0,2]=0.1, N[0,2]=1$;
		$S[1,2]=6, N[1,2]=10$;
		$S[2,2]=96, N[2,2]=160$.
		
		\textbf{Execution Trace}:
		
		1. \textbf{Top-level recursion}: $Rank(4, 2, \bm{x}, \emptyset)$
		Sequence decomposition: Prefix $\bm{p}_1=[4,9,4]$, Suffix $\bm{s}_1=[11]$.
		Suffix type is $\mathbb{C}_{spec}$, length $len=1$.
		\begin{itemize}
			\item Attempt $i=1$, Special class: Match! ($F=\emptyset \neq \mathbb{C}_{spec}$).
			\item Offset $V_{suf} = \text{EncodeSpecial}([11]) = 5$ (11 is the 6-th element in $\mathbb{C}_{spec}$, index 5).
			\item Recursively compute prefix rank $R_{pre} = Rank(3, 2, [4,9,4], \mathbb{C}_{spec})$.
			\item Contribution: $(R_{pre} - 1) \times 6^1 + 5 + 1$.
		\end{itemize}
		
		2. \textbf{Subproblem 1}: $Rank(3, 2, [4,9,4], \mathbb{C}_{spec})$
		Sequence decomposition: Prefix $\bm{p}_2=[4,9]$, Suffix $\bm{s}_2=[4]$.
		Suffix type is $\mathbb{C}_4$ (Normal), length $len=1$. Forbidden type $F=\mathbb{C}_{spec}$.
		\begin{itemize}
			\item Attempt $i=1$, Special class: $F=\mathbb{C}_{spec}$, Skip.
			\item Attempt $i=1$, Normal class $v=2$: Mismatch.
			\begin{itemize}
				\item Accumulate count: $count = S[2,2] + \frac{9}{10}N[2,2] = 96 + 0.9 \times 160 = 240$.
				\item $accum \leftarrow 240$.
			\end{itemize}
			\item Attempt $i=1$, Normal class $v=4$: Match!
			\begin{itemize}
				\item Recursively compute prefix $R_{pre2} = Rank(2, 2, [4,9], \mathbb{C}_4)$.
				\item Contribution: $240 + R_{pre2}$.
			\end{itemize}
		\end{itemize}
		
		3. \textbf{Subproblem 2}: $Rank(2, 2, [4,9], \mathbb{C}_4)$
		Sequence decomposition: Prefix $\bm{p}_3=[4]$, Suffix $\bm{s}_3=[9]$.
		Suffix type is $\mathbb{C}_{spec}$, length $len=1$. Forbidden type $F=\mathbb{C}_4$.
		\begin{itemize}
			\item Attempt $i=1$, Special class: Match!
			\item $V_{suf} = \text{EncodeSpecial}([9]) = 3$ (9 is the fourth element in $\mathbb{C}_{spec}$).
			\item Recursively compute $R_{pre3} = Rank(1, 2, [4], \mathbb{C}_{spec})$.
			\item Contribution: $0 + (R_{pre3} - 1) \times 6 + 3 + 1$.
		\end{itemize}
		
		4. \textbf{Subproblem 3}: $Rank(1, 2, [4], \mathbb{C}_{spec})$
		Suffix type is $\mathbb{C}_4$. Forbidden type $F=\mathbb{C}_{spec}$.
		\begin{itemize}
			\item Attempt $i=1$, Special class: Skip.
			\item Attempt $i=1$, Normal class $v=2$: Mismatch.
			\begin{itemize}
				\item $count = S[0] + 0.9 N[0] = 0.1 + 0.9 = 1$.
				\item $accum \leftarrow 1$.
			\end{itemize}
			\item Attempt $i=1$, Normal class $v=4$: Match!
			\begin{itemize}
				\item Contribution: $1 + Rank(0) = 1 + 1 = 2$.
			\end{itemize}
		\end{itemize}
		
		\textbf{Backtracking Computation}:
		\begin{itemize}
			\item $R_{pre3} = 2$.
			\item $R_{pre2} = (2-1) \times 6 + 4 = 10$.
			\item $R_{pre} = 240 + 10 = 250$.
			\item $Rank(4, 2, \bm{x}, \emptyset) = (250-1) \times 6 + 6 = 1494 + 6 = 1500$.
		\end{itemize}
		Therefore, the Rank of the sequence $[4, 9, 4, 11]$ is 1500.
	\end{example}
	
	\subsection{Unranking Algorithm (for Encoding)}
	
	We now describe the inverse of ranking, namely the \emph{unranking} algorithm. Unranking is the encoder's core subroutine: it reconstructs the generalized RLL sequence $\bm{x}$ from a given integer rank $M$.
	
	To achieve this, the algorithm explores all valid sequence branches following the exact same lexicographical traversal order established during the ranking phase (prioritizing suffix length $i$, followed by suffix types). Rather than accumulating counts to find a position, the unranking process iteratively compares the target rank $M$ against the precomputed sequence capacities of the current branch. If $M$ exceeds a branch's capacity, the algorithm bypasses that branch and updates the accumulated count. Once the correct branch containing $M$ is identified, the algorithm deduces the localized rank and recursively resolves the prefix structure.
	
	The complete recursive procedure for this reconstruction is formalized in Algorithm \ref{alg:unranking}, where the operator $\circ$ denotes sequence concatenation, and the notation $(v)^i$ within the algorithm represents a sequence consisting of the symbol $v$ repeated $i$ times.
	
	\begin{algorithm}[htbp]
		\caption{Unranking Algorithm: $Unrank(n, \ell, M, F)$}
		\label{alg:unranking}
		\begin{algorithmic}[1]
			\State \textbf{Input:} Length $n$, run-length limit $\ell$, rank $M$, forbidden type $F$
			\State \textbf{Output:} Sequence $\bm{x}$
			
			\If{$n=0$} \Return $\emptyset$ \EndIf
			\State $accum \leftarrow 0$
			
			\For{$i \leftarrow 1$ to $\ell$}
			\State \textit{// Attempt Special Class}
			\If{$F \neq \mathbb{C}_{spec}$}
			\State $W_{suf} \leftarrow 6^i$
			\State $W_{pre} \leftarrow N[n-i, \ell]$
			\State $count \leftarrow W_{suf} \times W_{pre}$
			
			\If{$M \le accum + count$}
			\State $M_{local} \leftarrow M - accum$
			\State $R_{pre} \leftarrow \lceil M_{local} / W_{suf} \rceil$
			\State $V_{suf} \leftarrow (M_{local} - 1) \pmod{W_{suf}}$
			
			\State $\bm{x}_{pre} \leftarrow Unrank(n-i, \ell, R_{pre}, \mathbb{C}_{spec})$
			\State $\bm{x}_{suf} \leftarrow \text{DecodeSpecial}(V_{suf}, i)$
			\State \Return $\bm{x}_{pre} \circ \bm{x}_{suf}$
			\EndIf
			\State $accum \leftarrow accum + count$
			\EndIf
			
			\State \textit{// Attempt Normal Classes}
			\For{$v \in \mathcal{N}$}
			\If{$F \neq \mathbb{C}_v$}
			\State $count \leftarrow S[n-i, \ell] + \frac{9}{10} N[n-i, \ell]$
			
			\If{$M \le accum + count$}
			\State $R_{pre} \leftarrow M - accum$
			\State $\bm{x}_{pre} \leftarrow Unrank(n-i, \ell, R_{pre}, \mathbb{C}_v)$
			\State \Return $\bm{x}_{pre} \circ (v)^i$
			\EndIf
			\State $accum \leftarrow accum + count$
			\EndIf
			\EndFor
			\EndFor
		\end{algorithmic}
	\end{algorithm}
	
	\subsection{Overall Encoding and Decoding Procedures} \label{subsec:overall_process}
	
	Leveraging the ranking and unranking algorithms, the complete encoding and decoding procedures for the generalized RLL codes are formulated as follows:
	
	\begin{itemize}
		\item \textbf{Encoding Procedure}: 
		
		\textbf{Input}: Code length $n$ and run-length limit $\ell$. Calculate the information capacity $m = \lfloor \log_{16} |\mathcal{C}(n,\ell)| \rfloor$. Let the information symbols be $\bm{x} = (x_1, \ldots, x_m) \in \Sigma_{16}^m$ (where each $x_i$ corresponds to a DNA pair).
		
		\textbf{Steps}:
		\begin{enumerate}
			\item Convert the information sequence $\bm{x}$ into an integer rank $M$. Since the unranking algorithm outputs a rank ranging from $1$ to $|\mathcal{C}|$, we treat $\bm{x}$ as a base-16 integer and add 1 to its value:
			\begin{equation*}
				M = 1 + \sum_{k=1}^{m} x_k \cdot 16^{m-k}
			\end{equation*}
			\item Invoke the unranking algorithm to generate the valid codeword:
			\begin{equation*}
				\bm{c} = Unrank(n, \ell, M, \emptyset)
			\end{equation*}
		\end{enumerate}
		\textbf{Output}: A codeword $\bm{c} \in \Sigma_{16}^n$ strictly satisfying the generalized RLL constraint.
		
		\item \textbf{Decoding Procedure}:
		
		\textbf{Input}: Code length $n$, run-length limit $\ell$, and the retrieved codeword $\bm{c} \in \mathcal{C}(n,\ell)$.
		
		\textbf{Steps}:
		\begin{enumerate}
			\item Invoke the ranking algorithm to calculate the rank $M$:
			\begin{equation*}
				M = Rank(n, \ell, \bm{c}, \emptyset)
			\end{equation*}
			\item Subtract 1 from $M$ and convert the result back into its $m$-digit base-16 representation, recovering the original information sequence $\bm{x}$.
		\end{enumerate}
		\textbf{Output}: The original information sequence $\bm{x}$.
	\end{itemize}
	
	\subsection{Complexity Analysis}\label{subsec:gen_RLL_complexity}
	
	By design, the encoded codewords strictly satisfy the generalized $\ell$-RLL constraint, and the decoding procedure losslessly recovers the original information. The algorithmic complexity is evaluated as follows:
	
	\begin{itemize}
		\item \textbf{Space Complexity}: The primary overhead stems from storing the precomputed counting tables $S[i, \ell]$ and $N[i, \ell]$ for $0 \le i \le n$. These tables contain $\mathcal{O}(n)$ entries. Because the values in the tables scale exponentially, the number of bits required to store each large integer grows linearly with $n$ (i.e., $\mathcal{O}(n)$ bits per entry). Consequently, the total space required for the counting tables is $\mathcal{O}(n^2)$. Furthermore, the recursion stack depth is bounded by $\mathcal{O}(n)$. Thus, the overall space complexity is strictly bounded by $\mathcal{O}(n^2)$.
		
		\item \textbf{Time Complexity}: Both the ranking and unranking operations inherently involve at most $n$ levels of recursion. During each recursive step, the algorithm performs large integer addition, multiplication, or division operations. Given that the precision of these integers is $\mathcal{O}(n)$ bits, each arithmetic operation takes at most $\mathcal{O}(n)$ time. Therefore, the total time complexity evaluates to $n \times \mathcal{O}(n) = \mathcal{O}(n^2)$.
	\end{itemize}
	
	\section{Construction of Burst $t$-Deletion/Insertion $\mathcal{A}$-Labeling Codes}
	\label{sec:construction}
	
	We now present the explicit construction. The key idea is to interleave the codeword into a $t\times(n/t)$ matrix so that a burst of $t$ consecutive label-domain indels becomes \emph{one} single-row indel in each row. Cascaded VT/SVT decoding then corrects all rows, provided that the burst can be coarsely localized; the latter is achieved by a label-domain $L$-RLL constraint enforced through the generalized RLL constraint of Section~\ref{sec:generalized_rll} together with a dedicated separator pattern between the information and parity regions.
	
	\subsection{Matrix Mapping and Error Model}

    Throughout this section we assume $t \mid n$ and $n \ge 7t + 3$ for simplicity of presentation.
    
    \textbf{\textit{Remark:}} If the target code length $n$ is not a multiple of $t$, one can logically pad the DNA sequence with up to $t-1$ predefined dummy symbols to reach an effective length of $n' = \lceil n/t \rceil t$. The encoding and decoding procedures then operate on this padded sequence. Because the burst length $t$ is a constant independent of $n$, this padding operation introduces an additional redundancy of at most $t-1$. This extra overhead is strictly absorbed into the $\mathcal{O}(1)$ term, thereby fully preserving the asymptotic optimality of the proposed construction for any arbitrary $n$.
    
    The following column-major interleaver is the centrepiece of the construction.
	
	\begin{definition}[Matrix Mapping]
		\label{def:matrix_mapping}
		Define the bijection $\mathcal{M}_{t}: \Sigma_q^n \to \Sigma_q^{t \times \frac{n}{t}}$ by $\bm{X}_{i,j} = x_{(j-1)t + i}$ for $1\le i\le t$ and $1\le j\le n/t$. Equivalently,
		\begin{equation*}
			\bm{X}=\begin{bmatrix}
				x_1 & x_{t+1} & \cdots &x_{n-t+1}\\
				x_2 & x_{t+2} & \cdots &x_{n-t+2}\\
				\vdots & \vdots & \ddots & \vdots\\
				x_{t} & x_{2t} &\cdots & x_{n}
			\end{bmatrix}.
		\end{equation*}
	\end{definition}
	
	For a codeword $\bm{x}\in\Sigma_4^n$, we write $\bm{Y} \triangleq \mathcal{M}_t(L_{\mathcal{A}}(\bm{x}))\in\Sigma_{11}^{t\times n/t}$ and define the composite mapping
	\begin{equation*}
		\Phi(\bm{X}) \triangleq \mathcal{M}_t\big( L_{\mathcal{A}}\big(\mathcal{M}_t^{-1}(\bm{X})\big)\big),
	\end{equation*}
	so that $\bm{Y}=\Phi(\bm{X})$. Note that $\Phi$ is not surjective, since $L_{\mathcal{A}}$ is not.
	
	\paragraph*{Error model in the matrix domain}
	A burst of $t$ consecutive deletions (resp.\ insertions) on $L_{\mathcal{A}}(\bm{x})$ induces a structured error pattern in $\bm{Y}$:
	\begin{itemize}
		\item Each row loses (resp.\ gains) exactly one entry.
		\item The column indices of the affected entries in rows $i$ and $i+1$ differ by at most $1$.
	\end{itemize}
	This is the key reduction underlying our scheme: a single $t$-burst in $1$D becomes $t$ single-row indels in $2$D, whose locations are highly correlated across rows.
	
	\subsection{Explicit Construction}
	
	To correct the matrix-domain errors described above, we design a target label matrix set $\mathcal{Y}$ such that every $\bm{Y}\in\mathcal{Y}$ is self-localizing and self-checking. We first fix the dimensions that govern the codeword layout, then state the algebraic constraints.

    \subsubsection{Parameter Configurations}\label{def:mdef}
	
	The construction begins by determining the optimal boundary between the information and redundancy regions. Specifically, to guarantee that the parity region provides sufficient storage capacity for all required checksum parameters, we define the number of columns in the information part, denoted as $m$, as the maximum positive integer satisfying the following inequality:
	\begin{equation} \label{eq:m_constraint}
		\left\lceil \frac{r_d - (t-1)}{t} \right\rceil + m + 3 \le \frac{n}{t},
	\end{equation}
	where the required number of parity symbols $r_d$ is defined as:
	\begin{align*}
		r_d \triangleq 2 + \Biggl\lceil  t \cdot \log_4 22 - \log_4 2 + \log_4 m 
		+ (t-1) \cdot \log_4(P+1) \Biggr\rceil,
	\end{align*}
	and $P = \lceil \log_{8/3} m \rceil + 2$. It is straightforward to verify that a valid solution for $m$ is guaranteed to exist when $n \ge 7t + 3$.
	
	\subsubsection{Constraint Set $\mathcal{Y}$}
	
	With these dimensions fixed, we now define the matrix set $\mathcal{Y}\subset\Sigma_{11}^{t\times n/t}$. As illustrated in Fig.~\ref{fig:1}, every $\bm{Y}\in\mathcal{Y}$ is logically partitioned into three regions:
	
	\begin{figure}[htbp]
		\centering
		\includegraphics[width=0.8\linewidth]{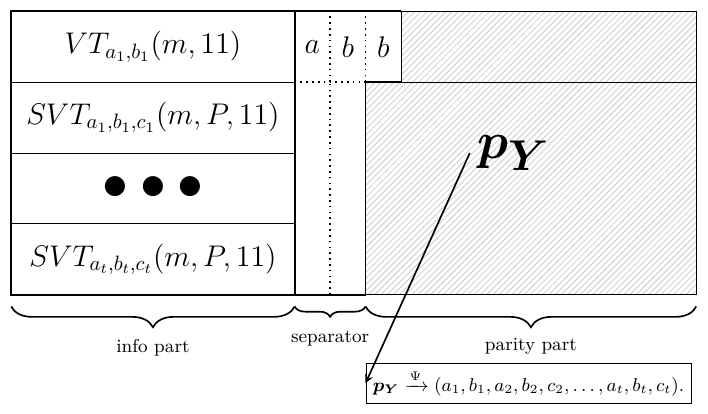}
		
		\caption{Schematic diagram of the structural requirements for the matrix set $\mathcal{Y}$. Note that 'a' and 'b' in the separator merely denote an inequality/equality relationship ($a \neq b$), rather than specific symbol values.}
		\label{fig:1}
	\end{figure}
	
	\begin{itemize}
		\item \textbf{Information Part (Info Part)}: $\bm{Y}_{[1:t], [1:m]}$, dedicated to carrying the information symbols.
		\item \textbf{Separator Part}: $\bm{Y}_{[1:t], [m+1:m+2]} \cup \{\bm{Y}_{1, m+3}\}$, serving as a robust barrier to distinguish whether a deletion or insertion error has occurred within the information part or the parity part.
		\item \textbf{Parity Part}: $\bm{Y}_{[2:t], m+3} \cup \bm{Y}_{[1:t], [m+4:n/t]}$, utilized for storing the redundancy required for the error correction algorithms.
	\end{itemize}
	
	Building upon this tripartite structure, we first serialize the parity part. Let $\bm{p_Y}$ denote the one-dimensional vector obtained by reading the parity region in column-major order:
	\begin{equation*}
		\bm{p_Y} = \left( \bm{Y}_{2,m+3}, \ldots, \bm{Y}_{t,m+3}, \bm{Y}_{1,m+4}, \ldots, \bm{Y}_{t,n/t} \right).
	\end{equation*}
	This parity vector is designed to store the target checksum parameters for the information rows. We can now rigorously define the structural constraints of $\mathcal{Y}$.
	
	\begin{definition}[Target Label Matrix Set $\mathcal{Y}$] \label{def:Y_set}
		Let $L = \lceil \log_{8/3} m \rceil$. The target label matrix set $\mathcal{Y} \subset \Sigma_{11}^{t \times \frac{n}{t}}$ comprises all matrices $\bm{Y}$ that simultaneously satisfy the following three conditions:
		\begin{enumerate}
			\item The first information row $\bm{Y}_{1,[1:m]}$ is an $L$-RLL sequence. 
			
			\item The elements in the separator part satisfy the relation:
			\begin{equation*}
				\bm{Y}_{1, m+1} \neq \bm{Y}_{1, m+2} = \bm{Y}_{1, m+3}.
			\end{equation*}
			
			\item There exists a deterministic injection $\Psi$ that decodes the parity vector $\bm{p_Y}$ into a sequence of target checksum parameters:
			\begin{equation*}
				\Psi(\bm{p_Y}) = (a_1, b_1, a_2, b_2, c_2, \ldots, a_t, b_t, c_t),
			\end{equation*}
			where $(a_1, b_1) \in \mathbb{Z}_m \times \mathbb{Z}_{11}$, and $(a_i, b_i, c_i) \in \mathbb{Z}_{P+1} \times \mathbb{Z}_{11} \times \mathbb{Z}_2$ for $2 \le i \le t$. Given these parameters, the information rows must satisfy:
			\begin{align*}
				\bm{Y}_{1,[1:m]} &\in \mathrm{VT}_{a_1, b_1}(m, 11), \\
				\bm{Y}_{i,[1:m]} &\in \mathrm{SVT}_{a_i, b_i, c_i}(m, P, 11), \quad \forall i \in [2, t].
			\end{align*}
		\end{enumerate}
	\end{definition}
	
	Based on the precisely defined constraint set $\mathcal{Y}$, we can now formally establish the construction of our code.
	
	\begin{construction}
		Let $n \ge 7t+3$ and $t \mid n$. The code $\mathcal{C}_{full} \subseteq \Sigma_4^n$ is defined as follows:
		\begin{equation*}
			\mathcal{C}_{full} = \left\{ \bm{x} \in \Sigma_4^n \mid \bm{Y}=\mathcal{M}_t(L_{\mathcal{A}}(\bm{x})) \in \mathcal{Y} \right\}.
		\end{equation*}
		In formal terms, $\mathcal{C}_{full}$ is the exact pre-image of the matrix set $\mathcal{Y}$ under the composite sequence mapping $\mathcal{M}_t(L_{\mathcal{A}}(\cdot))$.
	\end{construction}
	
	Section~\ref{sec:enc} describes an explicit encoder that maps arbitrary information symbols into a codeword of $\mathcal{C}_{full}$, and Section~\ref{sec:dec} shows that $\mathcal{C}_{full}$ is indeed a burst $t$-deletion/insertion $\mathcal{A}$-labeling code by giving a decoder that uniquely recovers the information from any single burst of length $t$ in the label domain.
	
	\section{Encoding Algorithm} \label{sec:enc}
	
	The encoder operates in two phases. The first phase places the information symbols into the Information Part of $\bm{X}$ so that the first row of the induced label matrix $\bm{Y}$ satisfies the $L$-RLL constraint. The second phase fills the Separator Part and the Parity Part of $\bm{X}$. We use the same Information/Separator/Parity partition for $\bm{X}$ and $\bm{Y}$.
	
	\subsection{Encoding the Information Part} \label{subsec:enc_info}
	
	Given the code length $n$ and the burst length $t$, the number of information columns $m$ is determined by Eq.~\eqref{eq:m_constraint}. First, we determine the maximum number of quaternary information symbols, $k$, that can be encoded.
	
	\subsubsection{Parameter Calculation}
	
	Recall that the required RLL constraint parameter is defined as $L = \lceil \log_{8/3} m \rceil$. To construct the information part, we use a generalized RLL code of length $m-2$ with a tighter constraint parameter of $L-1$.
	
	Let $k_{pre}$ be the maximum integer satisfying:
	\begin{equation}\label{eq:kpre_def}
		16^{k_{pre}} \le |\mathcal{C}(m-2, L-1)|,
	\end{equation}
	where the capacity $|\mathcal{C}(m-2, L-1)|$ is computed using the recurrence relations in Section \ref{subsec:gen_RLL_capacity}.
	
	The total number of quaternary information symbols is given by:
	\begin{equation}\label{eq:k_def}
		k \triangleq 2 k_{pre} + (t-2)m + 2.
	\end{equation}
	
	\subsubsection{Encoding Procedure}
	Given a sequence of $k$ information symbols $\bm{f} = (f_1, f_2, \ldots, f_k) \in \Sigma_4^k$, we utilize Algorithm~\ref{alg:encode_info} to populate the Info Part of the matrix $\bm{X}_{[1:t], [1:m]}$. The validity of Step~12 is guaranteed by Eq.~\eqref{eq:kpre_def}. A schematic diagram illustrating this mapping process is provided in Fig.~\ref{fig:enc}.
	
	\begin{figure}[htbp]
		\centering
		\includegraphics[width=1.0\linewidth]{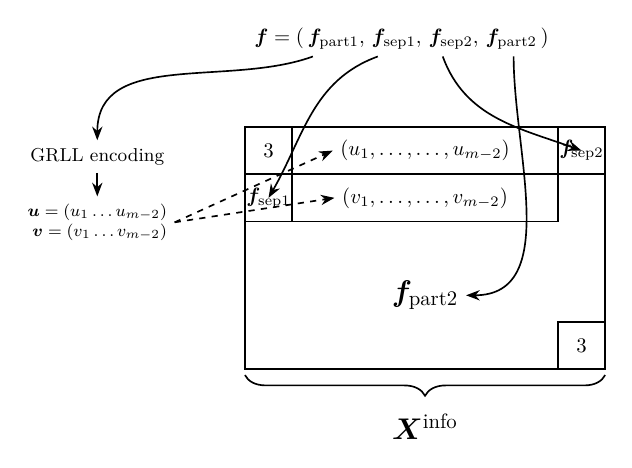}
		\caption{Schematic diagram illustrating the encoding algorithm for the Information Part of matrix $\bm{X}$.}
		\label{fig:enc}
	\end{figure}
	
	\begin{algorithm}[htbp]
		\caption{Encoding Algorithm for Info Part}
		\label{alg:encode_info}
		\begin{algorithmic}[1]
			\State \textbf{Input:} Code parameters $n, t, m$; Capacity parameters $k_{pre}, k$; Information sequence $\bm{f} \in \Sigma_4^k$
			\State \textbf{Output:} The encoded information submatrix $\bm{X}_{[1:t], [1:m]}$
			
			\State \textbf{Initialization}: Initialize all elements of $\bm{X}_{[1:t], [1:m]}$ to $0$.
			\State $\bm{X}_{1,1} \leftarrow 3$, $\bm{X}_{t,m} \leftarrow 3$
			\State Partition the input sequence $\bm{f}$ into four components $(\bm{f}_{\text{part1}}, \bm{f}_{\text{sep1}}, \bm{f}_{\text{sep2}}, \bm{f}_{\text{part2}})$, where:
			\Statex \quad $\bm{f}_{\text{part1}} \leftarrow (f_1, \ldots, f_{2k_{pre}})$
			\Statex \quad $\bm{f}_{\text{sep1}} \leftarrow f_{2k_{pre}+1}$
			\Statex \quad $\bm{f}_{\text{sep2}} \leftarrow f_{2k_{pre}+2}$
			\Statex \quad $\bm{f}_{\text{part2}} \leftarrow (f_{2k_{pre}+3}, \ldots, f_k)$
			\State $\bm{X}_{2,1} \leftarrow \bm{f}_{\text{sep1}}$, $\bm{X}_{1,m} \leftarrow \bm{f}_{\text{sep2}}$
			
			\If{$t > 2$}
			\State Define  $\mathcal{P} = \{\bm{X}_{2,m}\} \cup \bm{X}_{[3:t],[1:m-1]} \cup \bm{X}_{[3:t-1], m}$
			\State Sequentially place the $(t-2)m$ symbols from $\bm{f}_{\text{part2}}$ into the coordinates specified by $\mathcal{P}$.
			\EndIf
			
			\State Convert $\bm{f}_{\text{part1}}$ into a sequence of $k_{pre}$ symbols over $\Sigma_{16}$ utilizing the mapping $\psi$ from Definition \ref{def:DNApair}.
			\State Invoke the generalized RLL encoder  in Section~\ref{subsec:overall_process}:
			\Statex \quad \textit{Inputs}: Information $\bm{f}_{\text{part1}}$, target length $len = m-2$, constraint limit $\ell = L-1$
			\Statex \quad \textit{Output}: Constrained codeword sequence $\bm{c} \in \Sigma_{16}^{m-2}$
			
			\State For each symbol $c_j$ ($1 \le j \le m-2$) in $\bm{c}$, apply the inverse mapping $\psi^{-1}$ to recover the corresponding pair of $\Sigma_4$ symbols, i.e., $(u_j, v_j) = \psi^{-1}(c_j)$.
			\For{$j = 1$ to $m-2$}
			\State $\bm{X}_{1, j+1} \leftarrow u_j$
			\State $\bm{X}_{2, j+1} \leftarrow v_j$
			\EndFor
			
			\State \Return $\bm{X}_{[1:t], [1:m]}$
		\end{algorithmic}
	\end{algorithm}
	
	\subsubsection{Properties of the encoded Information Part}
	The following theorem records the two properties of Algorithm~\ref{alg:encode_info} that the decoder relies on.
	
	\begin{theorem} \label{thm:rll_property}
		The submatrix $\bm{X}_{[1:t],[1:m]}$ generated by Algorithm~\ref{alg:encode_info} satisfies the following properties:
		\begin{enumerate}
			\item \textbf{Fixed Boundaries}: $\bm{X}_{1,1} = \bm{X}_{t,m} = 3$.
			\item \textbf{RLL Satisfaction}: Let $\bm{Y} = \Phi(\bm{X})$. Regardless of the values assigned to the non-information part $\bm{X}_{[1:t], [m+1:n/t]}$, the first information row $\bm{Y}_{1, [1:m]}$ satisfies the $L$-RLL constraint.
		\end{enumerate}
	\end{theorem}
	
	\begin{IEEEproof}
		Property~1 follows directly from the assignments in Step~4 of Algorithm~\ref{alg:encode_info}. We now prove Property~2.
		
		For the first row of $\bm{Y}$, each element $\bm{Y}_{1,j}$ is determined by the corresponding pair $(\bm{X}_{1,j}, \bm{X}_{2,j})$. Therefore, $\bm{Y}_{1, [1:m]}$ depends only on the first two rows and the first $m$ columns of $\bm{X}$, and is independent of the non-information part $\bm{X}_{[1:t], [m+1:n/t]}$.
		
		In the encoding process, $\bm{f}_{\text{part1}}$ is mapped to a sequence $\bm{c} \in \Sigma_{16}^{m-2}$ that satisfies the generalized $(L-1)$-RLL constraint. The sequence $\bm{c}$ is then converted back into symbol pairs, which are assigned to $\bm{X}_{1,j}$ and $\bm{X}_{2,j}$ for $j=2, \ldots, m-1$.
		
		By Theorem \ref{thm:gen_rll_implication}, the inner sequence $\bm{Y}_{1, [2:m-1]}$ satisfies the $(L-1)$-RLL constraint. This means the maximum run length of identical symbols within this segment is at most $L-1$.
		
		The complete sequence $\bm{Y}_{1, [1:m]}$ is formed by appending the boundary elements $\bm{Y}_{1,1}$ and $\bm{Y}_{1,m}$ to the left and right of $\bm{Y}_{1, [2:m-1]}$, respectively.
		
		Since the maximum run length in the interior segment is $L-1$, these boundary elements can increase the overall run length to at most $(L-1) + 1 = L$. Thus, $\bm{Y}_{1, [1:m]}$ satisfies the $L$-RLL constraint.
	\end{IEEEproof}
	
	\subsection{Encoding the Parity Part}
	
	The parity part stores the critical parameters of the VT and SVT codes derived from the information part, which are utilized to recover the lost data when a burst deletion strikes the information section. The set of parameters that must be stored is denoted as $\mathcal{S}_{par} \triangleq (a_1, b_1, a_2, b_2, c_2, \ldots, a_t, b_t, c_t)$, where $a_i, b_i, c_i$ are defined in Definition~\ref{def:Y_set}.
	
	\subsubsection{Parameter Space and Mapping}
	
	The domains of these parameters are as follows:
	\begin{itemize}
		\item $a_1 \in[0, m-1]$, yielding $m$ possible values;
		\item $b_i \in[0, 10]$, yielding $11$ possible values ($1 \le i \le t$);
		\item $a_i \in [0, P]$, yielding $P+1$ possible values ($2 \le i \le t$);
		\item $c_i \in [0, 1]$, yielding $2$ possible values ($2 \le i \le t$).
	\end{itemize}
	The total number of possible parameter combinations, $N_{total}$, is evaluated as:
	\begin{equation*}
		N_{total} = m \cdot 11 \cdot \prod_{i=2}^t \left( (P+1) \cdot 11 \cdot 2 \right) = 11^t \cdot m \cdot (2P+2)^{t-1}.
	\end{equation*}
	
	Based on the definition of $r_d$, we have $r_d - 2 = \lceil \log_4 N_{total} \rceil$. This indicates that we can injectively map the parameter combinations into a quaternary sequence of length $r_d-2$.
	
	We define a mapping $\Psi_{int}: \mathcal{S}_{par} \to \mathbb{Z}_{N_{total}}$, which utilizes a mixed-radix representation~\cite{donald1999art} to convert the parameter tuple into a single integer $V$. Let the parameter sequence $\mathbf{v}$ and its corresponding radix sequence $\mathbf{M}$ be:
	\begin{align*}
		\mathbf{v} &= (a_1, b_1, a_2, b_2, c_2, \ldots, a_t, b_t, c_t), \\
		\mathbf{M} &= (m, 11, P+1, 11, 2, \ldots, P+1, 11, 2).
	\end{align*}
	The recursive formula to compute the integer value $V$ is:
	\begin{equation*}
		V = (\ldots((a_1 \cdot 11 + b_1) \cdot (P+1) + a_2) \cdot 11 + b_2) \cdot 2 + c_2 \ldots 
	\end{equation*}
	More formally, assuming $\mathbf{v}$ has a length of $K$, then $V = \sum_{i=1}^K v_i \cdot \left( \prod_{j=i+1}^K M_j \right)$ (where an empty product is treated as 1).
	
	Next, the integer $V$ is converted into a base-4 sequence of length $r_d-2$, denoted as $\bm{s}_{code} \in \Sigma_4^{r_d-2}$.
	Finally, to mark the boundaries of the parity information, we append a fixed symbol $3$ (representing the DNA base G) to both ends of $\bm{s}_{code}$. This yields the final parity sequence $\bm{a}_{parity}$:
	\begin{equation*}
		\bm{a}_{parity} \triangleq (3, \bm{s}_{code}, 3) \in \Sigma_4^{r_d}.
	\end{equation*}
	
	\subsubsection{Populating the Parity Matrix}
	
	We define the serialized vector $\bm{p_X}$ for the parity part of matrix $\bm{X}$ as:
	\begin{align*}
		\bm{p_X} = ( \bm{X}_{2,m+3}, \ldots, \bm{X}_{t,m+3}, 
		\bm{X}_{1,m+4}, \ldots, \bm{X}_{t,m+4}, \ldots, 
		&\bm{X}_{1,n/t}, \ldots, \bm{X}_{t,n/t} ).
	\end{align*}
	This implies reading all remaining elements starting from $\bm{X}_{2,m+3}$ in a column-major order.
	
	According to Eq.~\eqref{eq:m_constraint}, the total capacity of the parity region, $|\bm{p_X}| = (t-1) + t \cdot (n/t - m - 3)$, is at least $r_d$. We place the generated sequence $\bm{a}_{parity}$ into the first $r_d$ positions of $\bm{p_X}$, and fill all remaining positions in $\bm{p_X}$ with the symbol \textbf{1} (representing the DNA base T).
	
	\subsubsection{Proof of Parameter Recoverability}
	
	We must formally prove Condition 3 of the matrix set $\mathcal{Y}$, which asserts the existence of a fixed injection $\Psi$ capable of recovering the parameter set $\mathcal{S}_{par}$ from the labeled parity sequence $\bm{p}_{\bm{Y}}$. We formalize this recoverability in the following theorem, demonstrating that the encoding rules ensure a one-to-one mapping back to $\mathcal{S}_{par}$.
	
	\begin{theorem}\label{thm:pY_to_Spar}
		If $\bm{p}_{\bm{X}}$ is generated by the above encoding rules and $\bm{p}_{\bm{Y}} = L_{\mathcal{A}}(\bm{p}_{\bm{X}})$, then the parameter set $\mathcal{S}_{par}$ can be uniquely recovered from $\bm{p}_{\bm{Y}}$.
	\end{theorem}
	
	\begin{IEEEproof}
		According to the encoding rules, the structure of $\bm{p}_{\bm{X}}$ consists of the  data $\bm{a}_{parity}$ followed by a certain number of padding symbols $1$. Since $\bm{a}_{parity}$ strictly ends with a $3$, the tail pattern of $\bm{p}_{\bm{X}}$ is inevitably $(\dots, 3, 1, 1, \dots, 1)$ or simply $(\dots, 3)$ (when there is no padding).
		
		In the label domain $\Sigma_{11}$ (based on the label set $\mathcal{S}$):
		\begin{itemize}
			\item The DNA pair $(3, 1)$ (i.e., GT) corresponds to the label \textbf{6}.
			\item The DNA pair $(1, 1)$ (i.e., TT) corresponds to the label \textbf{10}.
		\end{itemize}
		We analyze the final symbols of $\bm{p}_{\bm{Y}}$ to determine the exact structure of $\bm{p}_{\bm{X}}$. There are two possible scenarios:
		\begin{enumerate}
			\item \textbf{Case 1: Padding exists ($r_d < |\bm{p}_{\bm{X}}|$).} 
			In this scenario, the tail of $\bm{p}_{\bm{Y}}$ will appear as $(\dots, 6, 10, \dots, 10, 0)$ or $(\dots, 6, 0)$. We can determine the exact length of the padding region by finding the continuous $10$s preceded by a $6$ at the end. This allows us to accurately extract the valid data portion.
			
			\item \textbf{Case 2: No padding ($r_d = |\bm{p}_{\bm{X}}|$).} 
			Here, $\bm{p}_{\bm{X}}$ ends exactly with $3$. The last valid label in $\bm{p}_{\bm{Y}}$ corresponds to the pair formed by the second-to-last element of $\bm{a}_{parity}$ and the terminal $3$. Because $\bm{p}_{\bm{X}}$ ends with $3$ (G), this final pair cannot be $(1, 1)$ (TT) or $(3, 1)$ (GT). Therefore, the pattern of ``a $6$ followed by several $10$s'' from Case 1 is a unique indicator of padding. The absence of this pattern clearly shows that no padding occurred.
		\end{enumerate}
		
		We can employ Algorithm \ref{alg:recover_parity} to extract the effective labeled sequence $L_{\mathcal{A}}(\bm{a}_{parity})$ from $\bm{p}_{\bm{Y}}$, where $\bm{p_Y}[cur]$ denotes the element at index $cur$ in $\bm{p_Y}$ (assuming 1-based indexing).
		
		\begin{algorithm}[htbp]
			\caption{Recover Parity Labels}
			\label{alg:recover_parity}
			\begin{algorithmic}[1]
				\State \textbf{Input:} The labeled sequence of the parity part $\bm{p_Y}$
				\State \textbf{Output:} The labeled sequence of the effective parity data $\bm{y}_{target} = L_{\mathcal{A}}(\bm{a}_{parity})$
				
				\State $cur \leftarrow |\bm{p_Y}| - 1$ 
				\If{$\bm{p_Y}[cur] == 10$}
				\While{$\bm{p_Y}[cur] == 10$}
				\State $cur \leftarrow cur - 1$
				\EndWhile
				\ElsIf{$\bm{p_Y}[cur] == 6$}
				\State $cur \leftarrow |\bm{p_Y}| - 1$
				\Else
				\State $cur \leftarrow |\bm{p_Y}|$
				\EndIf
				
				\State $\bm{y}_{target} \leftarrow$ the first $cur$ elements of $\bm{p_Y}$
				\State \Return $\bm{y}_{target}$
			\end{algorithmic}
		\end{algorithm}
		
		The sequence obtained through the algorithm above is exactly $L_{\mathcal{A}}(\bm{a}_{parity})$. 
		
		The parameter recovery process proceeds as follows:
		\begin{enumerate}
			\item \textbf{Reconstruct $\bm{a}_{parity}$}: Since the boundary elements of $\bm{a}_{parity}$ are a priori known to be both $3$, we can invoke Theorem \ref{thm:reconstruction} to uniquely decode the original sequence $\bm{a}_{parity}$ from $L_{\mathcal{A}}(\bm{a}_{parity})$.
			\item \textbf{Extract $\bm{s}_{code}$}: Remove the boundary symbols $3$ from both ends of $\bm{a}_{parity}$ to extract $\bm{s}_{code}$.
			\item \textbf{Recover Parameters}: Convert the quaternary sequence $\bm{s}_{code}$ back into the integer $V$, and reversely execute the mixed radix representation (successively applying modulo and integer division operations using the radices in $\mathbf{M}$) to uniquely reconstruct the parameter tuple $\mathcal{S}_{par}$.
		\end{enumerate}
	\end{IEEEproof}
	
	\subsection{Encoding the Separator Part}
	
	The separator part marks the boundary between the information and parity parts, helping the decoder determine whether an insertion or deletion error occurred in the information region or the parity region. The encoding of this part is straightforward. Its main goal is to insert a fixed pattern into $\bm{X}$ so that the corresponding labels in $\bm{Y}$ satisfy the separator condition: $\bm{Y}_{1, m+1} \neq \bm{Y}_{1, m+2} = \bm{Y}_{1, m+3}$.
	
	We simply set the values of matrix $\bm{X}$ in the separator region as follows:
	\begin{equation*}
		\bm{X}_{1, m+1} = 1, \quad \bm{X}_{1, m+2} = 0, \quad \bm{X}_{1, m+3} = 0.
	\end{equation*}
	All other elements in these three separator columns (i.e., from row 2 down to $t$) are set to 0:
	\begin{equation*}
		\bm{X}_{i,j} = 0, \quad \forall i \in [2:t], j \in \{m+1, m+2, m+3\}.
	\end{equation*}
	
	\begin{theorem}\label{thm:separator_neq}
		The above separator-encoding rule ensures that the induced label matrix $\bm{Y} = \Phi(\bm{X})$ satisfies the separator condition:
		\begin{equation*}
			\bm{Y}_{1, m+1} \neq \bm{Y}_{1, m+2} = \bm{Y}_{1, m+3}.
		\end{equation*}
	\end{theorem}
	
	\begin{IEEEproof}
		The label value of $\bm{Y}_{1,j}$ is strictly determined by the DNA pair $(\bm{X}_{1,j}, \bm{X}_{2,j})$. Based on our explicit configuration:
		
		\begin{itemize}
			\item The pair in column $m+1$ is $(\bm{X}_{1, m+1}, \bm{X}_{2, m+1}) = (1, 0)$, corresponding to the DNA dinucleotide TA. In the label set $\mathcal{S}$, TA corresponds to the label value $7$. Thus, $\bm{Y}_{1, m+1} = 7$.
			
			\item The pair in column $m+2$ is $(\bm{X}_{1, m+2}, \bm{X}_{2, m+2}) = (0, 0)$, corresponding to the DNA dinucleotide AA. Since AA is completely absent from the label set $\mathcal{S}$, this pair maps to the mismatch label $0$. Thus, $\bm{Y}_{1, m+2} = 0$.
			
			\item The pair in column $m+3$ is $(\bm{X}_{1, m+3}, \bm{X}_{2, m+3}) = (0, 0)$. Following the identical logic, it corresponds to AA and maps to the mismatch label $0$. Thus, $\bm{Y}_{1, m+3} = 0$.
		\end{itemize}
		
		Consequently, $\bm{Y}_{1, m+1} = 7 \neq 0 = \bm{Y}_{1, m+2} = \bm{Y}_{1, m+3}$, which gives the pattern $(7,0,0)$ and fulfils the separator condition $\bm{Y}_{1,m+1} \neq \bm{Y}_{1,m+2} = \bm{Y}_{1,m+3}$.
	\end{IEEEproof}

\begin{example}\label{ex:encoding}
To intuitively illustrate the encoding process, consider a system with code length $n=20$. According to Eq.~\eqref{eq:m_constraint}, the number of information columns evaluates to $m=3$. Suppose the input information sequence is $\bm{f} = (0, 2, 3, 0)$. 

\begin{itemize}[leftmargin=*]
    \item \textbf{Step 1: Encoding the Information Submatrix.}
    According to Algorithm~\ref{alg:encode_info}, the input sequence $\bm{f}$ is first partitioned into four components based on the capacity constraint $k_{pre}=1$ and $t=2$: 
    \begin{itemize}
        \item $\bm{f}_{part1} = (f_1, f_2) = (0, 2)$
        \item $\bm{f}_{sep1} = f_3 = 3$
        \item $\bm{f}_{sep2} = f_4 = 0$
        \item $\bm{f}_{part2} = \emptyset$ (since $t=2$, $(t-2)m = 0$)
    \end{itemize}
    Next, $\bm{f}_{part1}$ is mapped to a $\Sigma_{16}$ sequence (yielding a raw hex sequence of $(2)$) and passed to the generalized RLL encoder, which outputs the constrained sequence $\bm{c} = (3)$. This symbol is inversely mapped back to DNA pairs: $c_1=3 \Rightarrow (u_1, v_1) = (0, 3)$.

    These pairs, along with the separator symbols $f_{sep1}$ and $f_{sep2}$ (assigned to $\bm{X}_{2,1}$ and $\bm{X}_{1,m}$ respectively), and the fixed boundary `3`s ($\bm{X}_{1,1}=3, \bm{X}_{t,m}=3$), are placed into the $2 \times 3$ information submatrix $\bm{X}_{info}$ column by column. The corresponding label matrix $\bm{Y}_{info}$ is generated strictly based on the vertical and horizontal pairings of these bases:
    \begin{equation*}
        \bm{X}_{info} = \begin{bmatrix} 3 & 0 & 0 \\ 3 & 3 & 3 \end{bmatrix}, \quad \bm{Y}_{info} = \begin{bmatrix} 5 & 0 & 0 \\ 3 & 3 & 6 \end{bmatrix}.
    \end{equation*}

    \item \textbf{Step 2: Calculating Checksum Parameters and Matrix Assembly.}
    We compute the VT and SVT parameters for the rows of $\bm{Y}_{info}$:
    \begin{itemize}
        \item \textbf{Row 1 ($\bm{Y}_{1,[1:3]} = [5, 0, 0]$):} The VT parameters evaluate to $a_1 = 2, b_1 = 5$.
        \item \textbf{Row 2 ($\bm{Y}_{2,[1:3]} = [3, 3, 6]$):} The SVT parameters evaluate to $a_2 = 0, b_2 = 1, c_2 = 1$.
    \end{itemize}
    This generates the exact parameter tuple $\mathcal{S}_{par} = (2, 5, 0, 1, 1)$. This tuple is serialized into a base-4 sequence. Accommodating the boundary `3`s, we obtain the length-9 parity sequence $\bm{p_X} = (3, 1, 2, 1, 1, 1, 2, 0, 3)$. 
    
    The separator columns ($j=4,5,6$) are assigned to enforce the condition $\bm{Y}_{1,4} \neq \bm{Y}_{1,5} = \bm{Y}_{1,6}$, while $\bm{p_X}$ is placed into the parity region starting at $\bm{X}_{2,6}$. This yields the final $2 \times 10$ DNA matrix $\bm{X}$:
    \begin{equation*}
        \bm{X} = \begin{bmatrix} 3 & 0 & 0 & 1 & 0 & 0 & 1 & 1 & 1 & 0 \\ 3 & 3 & 3 & 0 & 0 & 3 & 2 & 1 & 2 & 3 \end{bmatrix},
    \end{equation*}
    which translates to the DNA strands (Row 0: \texttt{GAATAATTTA}, Row 1: \texttt{GGGAAGCTCG}, mapped via $0=\mathrm{A}, 1=\mathrm{T}, 2=\mathrm{C}, 3=\mathrm{G}$).
\end{itemize}
\end{example}
	
	\subsection{Correctness and Complexity of the Encoder}
	
	The following theorem combines the three parts above and bounds the encoder's complexity.
	
	\begin{theorem}
		For any arbitrary sequence of $k$ quaternary information symbols, the codeword $\bm{x} = \mathcal{M}_{t}^{-1}(\bm{X})$ generated by the proposed encoding algorithm belongs to the defined code $\mathcal{C}_{full}$. Furthermore, the overall complexity of this encoding procedure is $\mathcal{O}(n^2)$.
	\end{theorem}
	
	\begin{IEEEproof}
		\begin{enumerate}
			\item \textbf{Proof of Correctness ($\bm{x} \in \mathcal{C}_{full}$)}:
			
			By definition, $\mathcal{C}_{full}$ consists of all sequences $\bm{x}$ whose corresponding label matrix $\bm{Y} = \Phi(\mathcal{M}_t(\bm{x}))$ belongs to the target set $\mathcal{Y}$. The set $\mathcal{Y}$ is defined by three critical constraints, which we verify as follows:
			\begin{enumerate}
				\item \textbf{First-Row RLL Constraint}: As established in Theorem \ref{thm:rll_property}, the mapped first row $\bm{Y}_{1, [1:m]}$ satisfies the $L$-RLL constraint.
				\item \textbf{Separator Pattern}: As demonstrated in Theorem \ref{thm:separator_neq}, the assigned values ensure that $\bm{Y}$ satisfies $\bm{Y}_{1, m+1} \neq \bm{Y}_{1, m+2} = \bm{Y}_{1, m+3}$.
				\item \textbf{Parameter Embedding}: As proven in Theorem \ref{thm:pY_to_Spar}, the parameter set $\mathcal{S}_{par}$ can be uniquely extracted from the corresponding label sequence $\bm{p_Y}$. This confirms the existence of the injection $\Psi$ such that $\bm{p_Y} \xrightarrow{\Psi} \mathcal{S}_{par}$.
			\end{enumerate}
			Hence the label matrix of the generated codeword satisfies all three defining conditions of $\mathcal{Y}$, so $\bm{x} \in \mathcal{C}_{full}$.
			
			\item \textbf{Proof of Complexity}:
			
			We analyze the time complexity of each encoding step:
			\begin{itemize}
				\item \textbf{Information Part}: 
				\begin{itemize}
					\item Invoking the generalized RLL encoding algorithm requires $\mathcal{O}(m^2) = \mathcal{O}(n^2)$ time.
					\item Computing the VT and SVT checksums requires scanning the information submatrix (of size $t \times m \le n$). This involves only modular additions and takes $\mathcal{O}(n)$ time.
				\end{itemize}
				\item \textbf{Parity Part}:
				\begin{itemize}
					\item The parameter mapping $\Psi_{int}$ involves mixed-radix conversions for $\mathcal{O}(t)$ parameters. Since $t$ is a constant and the parameter values are bounded by $n$, this computation takes $\mathcal{O}(1)$ time.
					\item Populating the parity bits involves sequentially assigning values to the vector $\bm{p_X}$, taking $\mathcal{O}(n)$ time.
				\end{itemize}
				\item \textbf{Overall}: Aggregating these stages, the total complexity is dominated by the RLL encoder, resulting in an overall complexity of $\mathcal{O}(n^2)$.
			\end{itemize}
		\end{enumerate}
	\end{IEEEproof}
	
	\section{Decoding Algorithm} \label{sec:dec}
	
	We describe the decoder for burst deletions and burst insertions separately. In both cases, decoding has two phases: (i) recover the Information Part of $\bm{X}$ from the corrupted label matrix; (ii) invert the encoding pipeline to recover the user data $\bm{f}$.
	
	\subsection{Handling $t$-Burst Deletions}\label{subsec:burst_del}
	
	Suppose the received labeling sequence suffers a burst deletion of length $t$, denoted as $\bm{y}^{deleted}$. Given that the burst length $t$ satisfies $t \mid n$, we can leverage the matrix mapping $\mathcal{M}_t$ to reshape it into a $t \times (n/t - 1)$ matrix, denoted as $\bm{Y}^{deleted}$:
	\begin{equation*}
		\bm{Y}^{deleted} = \mathcal{M}_{t}(\bm{y}^{deleted}).
	\end{equation*}
	
	Due to the geometric projection of the burst deletion in the matrix domain, each row of $\bm{Y}^{deleted}$ is missing exactly one element compared to the original matrix $\bm{Y}$. The primary task in decoding is to localize the region where the deletion occurred, which determines whether the Information Part requires error correction. As illustrated in Fig. \ref{fig:dec}, we divide the comprehensive decoding process into five sequential stages.
	
	\begin{figure}[htbp]
		\centering
		\includegraphics[width=.8\linewidth]{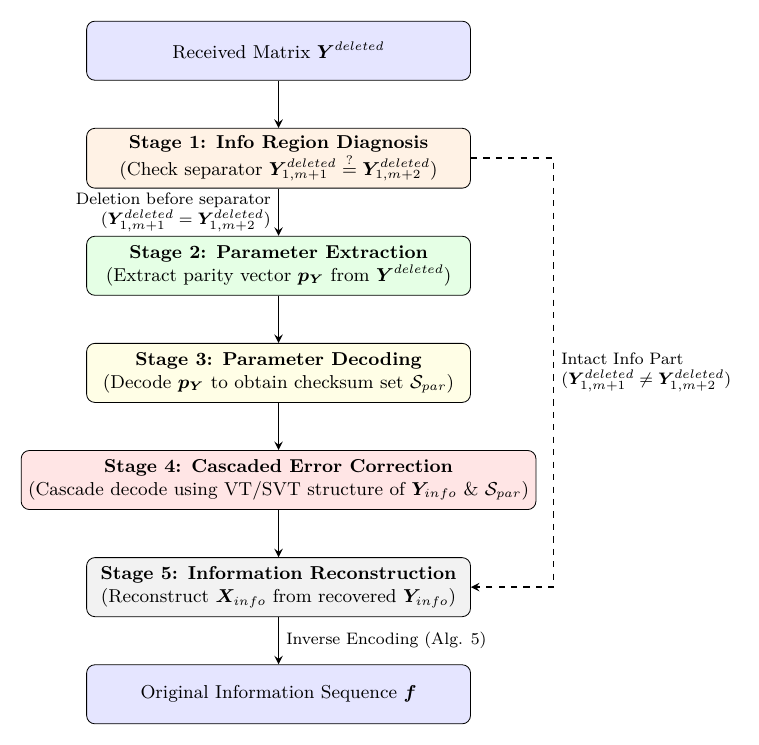}
		\caption{Flowchart of the comprehensive decoding algorithm.}
		\label{fig:dec}
	\end{figure}
	
	\subsubsection{Stage 1: Info Region Diagnosis}
	
	We inspect the elements in $\bm{Y}^{deleted}$ corresponding to the original separator positions by comparing the values of $\bm{Y}^{deleted}_{1, m+1}$ and $\bm{Y}^{deleted}_{1, m+2}$.
	
	\begin{itemize}
		\item \textbf{Case A (Intact Info Part): $\bm{Y}^{deleted}_{1, m+1} \neq \bm{Y}^{deleted}_{1, m+2}$}
		
		\textbf{Analysis}: If the received matrix preserves this inequality, the separator structure remains uncorrupted. This indicates that the entire $t$-burst deletion occurred in the sequence \emph{after} the separator $\bm{Y}_{1, m+1}$. In the matrix domain, the deleted element in each row is strictly located at column index $j \ge m+1$.
		
		\textbf{Action}: Because the Information Part remains fully intact, we directly extract the first $m$ columns of $\bm{Y}^{deleted}$ as our original information matrix: $\bm{Y}_{info} = \bm{Y}^{deleted}_{[1:t], [1:m]}$. We then bypass the error correction stages and proceed directly to \textbf{Stage 5}.
		
		\item \textbf{Case B (Deletion before separator): $\bm{Y}^{deleted}_{1, m+1} = \bm{Y}^{deleted}_{1, m+2}$}
		
		\textbf{Analysis}: If these two elements are equal, the original separator structure is corrupted. This indicates that the entire $t$-burst deletion occurred \emph{before} the separator $\bm{Y}_{1, m+2}$. The deleted element in each row is strictly located at column index $j \le m+1$, inducing a leftward shift of all subsequent elements.
		
		\textbf{Action}: We proceed to \textbf{Stage 2} to initiate the correction protocol.
	\end{itemize}
	
	\subsubsection{Stage 2: Parameter Extraction}
	
	For Case B, the parity part itself is structurally complete but shifted leftward. Therefore, we first extract the shifted first $m-1$ columns, denoting this corrupted information submatrix as $\bm{Y}' = \bm{Y}^{deleted}_{[1:t], [1:m-1]}$. Concurrently, we extract the parity vector $\bm{p_Y}$ from $\bm{Y}^{deleted}$ by reading from the original parity column indices shifted left by one coordinate.
	
	\subsubsection{Stage 3: Parameter Decoding}
	
	Based on Theorem \ref{thm:pY_to_Spar}, we decode the extracted parity vector $\bm{p_Y}$ to recover the exact VT and SVT checksum parameters $\mathcal{S}_{par}$ required for correcting the information part.
	
	\subsubsection{Stage 4: Cascaded Error Correction}
	
	Using the checksum parameters in $\mathcal{S}_{par}$, we correct the single deletion in each row of the corrupted matrix $\bm{Y}'$:
	\begin{itemize}
		\item \textbf{Recovering the First Row}: The first row $\bm{Y}_{1, [1:m]}$ satisfies the $L$-RLL constraint and constitutes an 11-ary VT code governed by parameters $(a_1, b_1) \in \mathcal{S}_{par}$. Employing the standard VT deletion-correcting algorithm~\cite{1984VT_def}, we uniquely recover $\bm{Y}_{1, [1:m]}$ from $\bm{Y}'_{1, [1:m-1]}$. Crucially, the $L$-RLL property restricts the ambiguity of the deletion location $j_1$ to an interval of length at most $L$.
		\item \textbf{Recovering the Remaining Rows}: Based on the geometric projection of a burst deletion, if the deletion in the first row occurred at index $j_1$, the deletion in the $i$-th row ($2 \le i \le t$) must fall within the neighborhood $[j_1-1, j_1+1]$. This bounds the potential deletion location to a window of length $P = L + 2$. Since the $i$-th row is an SVT code parameterized by $(a_i, b_i, c_i) \in \mathcal{S}_{par}$, we invoke the SVT decoding algorithm~\cite{2017q-ary-QVT_def} to recover the original $\bm{Y}_{i, [1:m]}$ from $\bm{Y}'_{i, [1:m-1]}$.
	\end{itemize}
	
	\subsubsection{Stage 5: Information Reconstruction}
	
	Through the preceding stages, we successfully obtain the fully recovered information matrix, $\bm{Y}_{info}$. Because the boundary elements of the info part in $\bm{X}$ were fixed to 3 during encoding ($\bm{X}_{1,1}=3, \bm{X}_{t,m}=3$), we apply the invertibility of the DNA label mapping (Theorem \ref{thm:reconstruction}). Relying on our recovered $\bm{Y}_{info}$ and these boundary conditions, we losslessly reconstruct the Info Part of $\bm{X}$, denoting it as $\bm{X}_{info} = \bm{X}_{[1:t], [1:m]}$.
	
	Once the Info Part $\bm{X}_{[1:t], [1:m]}$ is recovered, reconstructing the original information symbols $\bm{f}$ requires performing the inverse operations of the padding and encoding procedures, as outlined in Algorithm \ref{alg:info_recovery}.
	
	\begin{algorithm}[htbp]
		\caption{Info Symbol Recovery Algorithm}
		\label{alg:info_recovery}
		\begin{algorithmic}[1]
			\State \textbf{Input:} The recovered matrix information block $\bm{X}_{[1:t], [1:m]}$; parameters $m, t, k_{pre}$
			\State \textbf{Output:} The original information sequence $\bm{f}$
			
			\State Initialize an empty sequence $\bm{c} = (c_1, \ldots, c_{m-2}) \in \Sigma_{16}^{m-2}$
			\For{$j = 1$ to $m-2$}
			\State Extract symbols from corresponding columns in rows 1 and 2: $(u_j, v_j) \leftarrow (\bm{X}_{1, j+1}, \bm{X}_{2, j+1})$
			\State Convert the DNA pair to a $\Sigma_{16}$ symbol using mapping $\psi$: $c_j \leftarrow \psi(u_j, v_j)$
			\EndFor
			
			\State Invoke the generalized RLL decoder in Section~\ref{subsec:overall_process}:
			\Statex \quad \textit{Input}: Codeword $\bm{c}$
			\Statex \quad \textit{Output}: $k_{pre}$ information symbols over $\Sigma_{16}$, denoted as $\bm{f}_{part1}$
			\State Convert $\bm{f}_{\text{part1}}$ into $2k_{pre}$ symbols over $\Sigma_{4}$ using the inverse mapping $\psi^{-1}$.
			\State Extract edge symbols: $\bm{f}_{edge1} \leftarrow \bm{X}_{2,1}$, $\bm{f}_{edge2} \leftarrow \bm{X}_{1,m}$
			
			\If{$t > 2$}
			\State Define the target coordinate set $\mathcal{P} = \{\bm{X}_{2,m}\} \cup \bm{X}_{[3:t], [1:m-1]} \cup \bm{X}_{[3:t-1], m}$
			\State Extract $(t-2)m$ symbols from $\mathcal{P}$ adhering to the fixed order used in encoding to yield $\bm{f}_{part2}$.
			\Else
			\State $\bm{f}_{part2} \leftarrow \emptyset$
			\EndIf
			\State $\bm{f} \leftarrow (\bm{f}_{part1}, \bm{f}_{edge1}, \bm{f}_{edge2}, \bm{f}_{part2})$
			
			\State \Return $\bm{f}$
		\end{algorithmic}
	\end{algorithm}
	
	\subsection{Handling $t$-Burst Insertions} \label{subsec:burst_ins}
	
	The decoding procedure for burst insertions is symmetric to that of burst deletions. Suppose the received sequence suffers a $t$-burst insertion, denoted as $\bm{y}^{inserted}$. Using the matrix mapping $\mathcal{M}_t$, we reshape it into a $t \times (n/t + 1)$ matrix, $\bm{Y}^{inserted}$. Due to the geometric projection of the burst insertion, each row of $\bm{Y}^{inserted}$ contains exactly one additional element compared to the original matrix $\bm{Y}$.
	
	Because an insertion induces a rightward shift, we diagnose the error location by examining the right-side equality pattern of the separator, $\bm{Y}_{1, m+2} = \bm{Y}_{1, m+3}$. Mirroring the logic established in Section \ref{subsec:burst_del}, we divide the recovery process into the same five sequential stages.
	
	\subsubsection{Stage 1: Info Region Diagnosis}
	
	We inspect the separator region by comparing $\bm{Y}^{inserted}_{1, m+2}$ and $\bm{Y}^{inserted}_{1, m+3}$:
	\begin{itemize}
		\item \textbf{Case A (Intact Info Part): $\bm{Y}^{inserted}_{1, m+2} = \bm{Y}^{inserted}_{1, m+3}$.}
		
		\textbf{Analysis}: The equality pattern is preserved, confirming that the insertion occurred \emph{after} the Information Part.
		
		\textbf{Action}: We directly extract the first $m$ columns as our information matrix, $\bm{Y}_{info} = \bm{Y}^{inserted}_{[1:t], [1:m]}$, bypass the error correction stages, and proceed to \textbf{Stage 5}.
		
		\item \textbf{Case B (Insertion before separator): $\bm{Y}^{inserted}_{1, m+2} \neq \bm{Y}^{inserted}_{1, m+3}$.}
		
		\textbf{Analysis}: The separator structure is broken, indicating the insertion occurred \emph{prior} to the parity part (i.e., within the Information Part or separator), causing a rightward shift.
		
		\textbf{Action}: We proceed to \textbf{Stage 2} to initiate the correction protocol.
	\end{itemize}
	
	\subsubsection{Stage 2: Parameter Extraction}
	
	For Case B, the information submatrix has expanded. We extract this corrupted region as $\bm{Y}' = \bm{Y}^{inserted}_{[1:t], [1:m+1]}$. To accommodate the rightward shift, we extract the parity vector $\bm{p_Y}$ from the original parity column indices shifted right by one coordinate.
	
	\subsubsection{Stage 3: Parameter Decoding}
	
	Exactly as in the deletion process, we decode the extracted parity vector $\bm{p_Y}$ to recover the exact VT and SVT checksum parameters $\mathcal{S}_{par}$ required for error correction.
	
	\subsubsection{Stage 4: Cascaded Error Correction}
	
	VT and SVT codes can inherently correct both deletions and insertions. Leveraging $\mathcal{S}_{par}$, we correct the single insertion in each row of $\bm{Y}'$:
	\begin{enumerate}
		\item \textbf{Row 1 (VT Decoding):} Using parameters $(a_1, b_1) \in \mathcal{S}_{par}$, we apply the standard VT insertion-correcting algorithm~\cite{1984VT_def} to uniquely recover $\bm{Y}_{1, [1:m]}$. The $L$-RLL property restricts the ambiguity of the insertion location $j_1$ to a narrow interval.
		\item \textbf{Remaining Rows (SVT Decoding):} Due to the geometric projection of the burst, the insertion in the $i$-th row ($2 \le i \le t$) is tightly bounded to the neighborhood $[j_1-1, j_1+1]$. We independently invoke the SVT insertion decoding algorithm~\cite{2017QVT_def} for each row $i$ using parameters $(a_i, b_i, c_i) \in \mathcal{S}_{par}$ to recover $\bm{Y}_{i, [1:m]}$.
	\end{enumerate}
	
	\subsubsection{Stage 5: Information Reconstruction}
	
	As with the deletion scenario, the preceding stages yield the fully recovered information matrix, $\bm{Y}_{info}$. Because the subsequent reconstruction steps are identical, we apply the inverse DNA label mapping (Theorem \ref{thm:reconstruction}) using the known boundary conditions ($\bm{X}_{1,1}=3, \bm{X}_{t,m}=3$) to losslessly reconstruct the Info Part, yielding $\bm{X}_{info} = \bm{X}_{[1:t], [1:m]}$. Finally, we execute Algorithm~\ref{alg:info_recovery} on $\bm{X}_{info}$ to retrieve the original information sequence $\bm{f}$.

\begin{example}[Decoding Process]
Continuing from the encoded matrix $\bm{X}$ and its corresponding label matrix $\bm{Y}$ generated in Example~\ref{ex:encoding}, we demonstrate the cascaded decoding process when the sequence is subjected to a $t$-burst deletion and insertion with $t=2$. 

\begin{itemize}[leftmargin=*]
    \item \textbf{Handling a $t$-Burst Deletion (at index 1):} 
    Suppose a burst deletion of length $t=2$ strikes the labeling sequence within the information region. The received sequence is reshaped into a $2 \times 9$ matrix $\bm{Y}^{deleted}$:
    \begin{equation*}
        \bm{Y}^{deleted} = \begin{bmatrix} 5 & 0 & 7 & 0 & 0 & 8 & 10 & 8 & 0 \\ 3 & 6 & 0 & 0 & 6 & 0 & 10 & 2 & 0 \end{bmatrix}.
    \end{equation*}
    
    \begin{itemize}
        \item \textbf{Stage 1 (Info Region Diagnosis):} We evaluate the shifted separator condition. Checking columns $m+1=4$ and $m+2=5$, we find $\bm{Y}^{deleted}_{1, 4} = 0$ and $\bm{Y}^{deleted}_{1, 5} = 0$. Since $0 = 0$, the original separator pattern ($7 \neq 0 = 0$) is broken, signaling that the deletion occurred before the separator (Case B).
        
        \item \textbf{Stages 2 \& 3 (Parameter Extraction \& Decoding):} Acknowledging the leftward shift, we extract the parity sequence from the shifted parity region. The extracted vector corresponds to the original base-4 parity sequence $\bm{p_X} = (3, 1, 2, 1, 1, 1, 2, 0, 3)$. Decoding this yields the exact checksum parameters $\mathcal{S}_{par} = (2, 5, 0, 1, 1)$, meaning $(a_1,b_1)=(2,5)$ and $(a_2,b_2,c_2)=(0,1,1)$.
        
        \item \textbf{Stage 4 (Cascaded Error Correction):} We isolate the damaged information rows: $\bm{Y}'_{1,[1:2]} = [5, 0]$ and $\bm{Y}'_{2,[1:2]} = [3, 6]$. Using the VT decoder with $(a_1, b_1)$, we perfectly correct the missing label in Row 1 to reconstruct $[5, 0, 0]$, localizing the deletion error index. Propagating this positional bound to the SVT decoder with $(a_2,b_2,c_2)$, Row 2 is corrected to $[3, 3, 6]$.
    \end{itemize}

    \item \textbf{Handling a $t$-Burst Insertion (at index 1):} 
    Alternatively, suppose a burst insertion of length $t=2$ expands the sequence. The received labels are reshaped into a $2 \times 11$ matrix $\bm{Y}^{inserted}$:
    \begin{equation*}
        \bm{Y}^{inserted} = \begin{bmatrix} 5 & 4 & 0 & 0 & 7 & 0 & 0 & 8 & 10 & 8 & 0 \\ 8 & 3 & 3 & 6 & 0 & 0 & 6 & 0 & 10 & 2 & 0 \end{bmatrix}.
    \end{equation*}

    \begin{itemize}
        \item \textbf{Stage 1 (Info Region Diagnosis):} We check the insertion separator condition at columns $m+2=5$ and $m+3=6$. Here, $\bm{Y}^{inserted}_{1, 5} = 7$ and $\bm{Y}^{inserted}_{1, 6} = 0$. Since $7 \neq 0$, the equality pattern is destroyed (Case B), confirming the insertion occurred within the information region causing a rightward shift.
        
        \item \textbf{Stages 2, 3 \& 4 (Extraction \& Correction):} Similar to the deletion case, the shifted parity vector is extracted and decoded to $\mathcal{S}_{par} = (2, 5, 0, 1, 1)$. We isolate the corrupted info rows which now contain an extra symbol: $\bm{Y}'_{1,[1:4]} = [5, 4, 0, 0]$ and $\bm{Y}'_{2,[1:4]} = [8, 3, 3, 6]$. The VT and SVT insertion decoders identify the erroneous inserted labels (`4` in Row 1, `8` in Row 2) and remove them, successfully restoring the corrected rows $[5, 0, 0]$ and $[3, 3, 6]$.
    \end{itemize}
    
    \item \textbf{Stage 5 (Information Reconstruction):} 
    Both correction scenarios successfully restore the $2 \times 3$ information submatrix $\bm{Y}_{info}$ and its corresponding DNA matrix $\bm{X}_{info}$:
    \begin{equation*}
        \bm{X}_{info} = \begin{bmatrix} 3 & 0 & 0 \\ 3 & 3 & 3 \end{bmatrix}.
    \end{equation*}
    Following Algorithm~\ref{alg:info_recovery}, we extract the internal column pairs ($j=2$). The DNA pair $(u_1, v_1) = (0, 3)$ is mapped to the generalized RLL symbol $c_1 = 3$. Passing $\bm{c}=(3)$ through the generalized RLL decoder recovers the raw hex sequence $(2)$, which translates to $\bm{f}_{part1} = (0, 2)$. Finally, extracting the boundary symbols $f_{sep1} = \bm{X}_{2,1} = 3$ and $f_{sep2} = \bm{X}_{1,m} = 0$, we perfectly reconstruct the original quaternary sequence $\bm{f} = (0, 2, 3, 0)$.
\end{itemize}
\end{example}

	\subsection{Correctness and Complexity of the Decoder}
	
	We now verify that the decoder recovers $\bm{f}$ exactly for any single burst of length $t$, and bound its complexity.
	
	\begin{theorem}[Correctness and Complexity]
		The code generated by the encoding algorithm in Section~\ref{sec:enc} is a valid burst $t$-deletion/insertion $\mathcal{A}$-labeling code. Specifically, when subjected to a burst indel of length $t$, the proposed decoding algorithm successfully recovers the original information sequence $\bm{f}$. Furthermore, the  complexity of this decoding algorithm is $\mathcal{O}(n^2)$.
	\end{theorem}
	
	\begin{IEEEproof}
		The correctness of the algorithm is guaranteed by the deterministic nature of the decoding phases detailed above. The complexity is analyzed as follows:

		The primary computational overheads of the decoding process are:
		\begin{itemize}
			\item \textbf{Matrix Reshaping and Scanning}: Converting the 1D sequence into a 2D matrix and verifying the separator conditions requires a linear scan, executing in $\mathcal{O}(n)$ time.
			\item \textbf{Parameter Recovery}: Extracting parameters from $\bm{p_Y}$ involves a linear traversal and fundamental arithmetic operations (mixed-radix decoding), running in $\mathcal{O}(n)$ time.
			\item \textbf{VT/SVT Decoding}: The error correction algorithm for each individual row operates in linear time relative to the row length. With $t$ rows, each of length strictly less than $n/t + 1$, the cumulative overhead is bounded by $t \times \mathcal{O}(n/t) = \mathcal{O}(n)$.
			\item \textbf{Generalized RLL Decoding}: As the final step, mapping the constrained codeword $\bm{c}$ back to its integer rank constitutes the most computationally intensive operation. Based on the analysis in Section \ref{subsec:gen_RLL_complexity}, this step requires $\mathcal{O}(n^2)$ time.
		\end{itemize}
		Consequently, aggregating these phases reveals that the overall complexity of the decoding algorithm is dominated by the RLL decoder, resulting in $\mathcal{O}(n^2)$.
	\end{IEEEproof}
	
	\section{Redundancy Analysis and Asymptotic Optimality} \label{sec:redundancy}
	
	This section computes the redundancy of the proposed code and compares it to the lower bound of Theorem~\ref{thm:theoretical_bound_new}. We split the redundancy into two sources—the generalized RLL constraint on $\bm{X}$ and the matrix parity overhead—and bound them separately.
	
	With $r = n-k$, $n$ the code length over $\Sigma_4$, and $k$ as in Eq.~\eqref{eq:k_def}, we write:
	\begin{align}
		r &= n - k \nonumber \\
		&= n - (2k_{pre} + (t-2)m + 2) \nonumber \\
		&= \left( 2(m-2) - 2k_{pre} \right) + \left( n - tm + 2 \right) \nonumber \\
		&\triangleq r_1 + r_2. \label{eq:r_r1_r2}
	\end{align}
	Here:
	\begin{itemize}
		\item $r_1 = 2(m-2) - 2k_{pre}$ denotes the redundancy introduced by the generalized RLL encoding;
		\item $r_2 = n - tm + 2$ denotes the structural redundancy introduced by the Separator and Parity components within the matrix framework.
	\end{itemize}
	
	\subsection{Analysis of Generalized RLL Redundancy ($r_1$)}\label{subsec:r1_cal}
	
	We first evaluate $r_1$. From the definition of $k_{pre}$ (Eq.~\eqref{eq:kpre_def}), we have:
	\begin{equation*}
		k_{pre} \ge \log_{16} |\mathcal{C}(m-2, L-1)| - 1,
	\end{equation*}
	where $L = \lceil \log_{8/3} m \rceil$. According to Theorem \ref{thm:genRLL_redundancy}, as $m \to \infty$, the asymptotic code rate of the generalized RLL code is:
	\begin{equation*}
		r_{\mathcal{C}(m-2, L-1)} = \lim_{m \to \infty} \frac{\log_{16} |\mathcal{C}(m-2, L-1)|}{m-2} = \log_{16} \lambda,
	\end{equation*}
	where $\lambda$ is the largest positive real root of the characteristic equation:
	\begin{equation} \label{eq:charactize_eq}
		10 B\left(\frac{1}{\lambda}\right) A\left(\frac{1}{\lambda}\right) + 9 B\left(\frac{1}{\lambda}\right) = 1.
	\end{equation}
	Here, the auxiliary polynomials are defined as $A(x) = \sum_{i=1}^{L-1} 6^i x^i$ and $B(x) = \sum_{i=1}^{L-1} x^i$.
	
	Consequently, as $n \to \infty$ (and thus $m \to \infty$), the upper bound for $r_1$ is given by:
	\begin{align}
		r_1 &= 2(m-2) - 2k_{pre} \nonumber \\
		&\le 2(m-2) - 2\left( \log_{16} |\mathcal{C}(m-2, L-1)| - 1 \right) \nonumber \\
		&= 2(m-2) \left( 1 - \frac{\log_{16} |\mathcal{C}(m-2, L-1)|}{m-2} \right) + 2 \nonumber \\
		&= 2(m-2) \left( 1 - \log_{16} \lambda \right) + 2 \nonumber \\
		&= 2(m-2) \log_{16} \frac{16}{\lambda} + 2. \label{eq:r1_ineq}
	\end{align}

    Next, we estimate $\lambda$. Evaluating the geometric series for $A(x)$ and $B(x)$ yields:
	\begin{align*}
		A\left(\frac{1}{\lambda}\right) &= \frac{6}{\lambda-6} \left( 1 - \left(\frac{6}{\lambda}\right)^{L-1} \right), \\
		B\left(\frac{1}{\lambda}\right) &= \frac{1}{\lambda-1} \left( 1 - \left(\frac{1}{\lambda}\right)^{L-1} \right).
	\end{align*}
	Substituting these into the characteristic Eq.~\eqref{eq:charactize_eq} and rearranging:
	\begin{equation*}
		\frac{1 - \lambda^{-(L-1)}}{\lambda - 1} \left( 60 \frac{1 - (6/\lambda)^{L-1}}{\lambda - 6} + 9 \right) = 1.
	\end{equation*}
	As $n \to \infty$, it follows that $L \to \infty$, and one can verify that $\lambda \to 16$. Let $\lambda = 16 - \delta$, where $\delta = o(1)$. Using the asymptotic expansions for $L \to \infty$, we have $\lambda^{-(L-1)} = \mathcal{O}(16^{-L})$ and $(6/\lambda)^{L-1} = (3/8)^{L-1}(1 + o(1))$. The characteristic equation can be rewritten  as:
	\begin{align*}
		15 - \delta &= (1 - \lambda^{-(L-1)}) \left( \frac{60(1 - (6/\lambda)^{L-1})}{10 - \delta} + 9 \right).
	\end{align*}
	By expanding the terms involving $\delta$ with $(10 - \delta)^{-1} = \frac{1}{10}(1 + \frac{\delta}{10} + \mathcal{O}(\delta^2))$, we obtain:
	\begin{align*}
		15 - \delta &= (1 - \mathcal{O}(16^{-L})) 
		\times \Biggl[ 6 \left(1 - \left(\frac{3}{8}\right)^{L-1} \!\!+ o\big((3/8)^L\big)\right) 
		\times \left(1 + \frac{\delta}{10} + \mathcal{O}(\delta^2)\right) + 9 \Biggr] \\
		&= (1 - \mathcal{O}(16^{-L}))  \times \Biggl( 15 + \frac{3}{5}\delta - 6\left(\frac{3}{8}\right)^{L-1}  + o\big((3/8)^L\big) + \mathcal{O}(\delta^2) \Biggr).
	\end{align*}
	Since $\mathcal{O}(16^{-L}) = o((3/8)^L)$, multiplying out the right-hand side yields:
	\begin{equation*}
		15 - \delta = 15 + \frac{3}{5}\delta - 6\left(\frac{3}{8}\right)^{L-1} + o\big((3/8)^L\big) + \mathcal{O}(\delta^2).
	\end{equation*}
	Canceling $15$ from both sides and rearranging gives:
	\begin{equation*}
		\frac{8}{5}\delta = 6\left(\frac{3}{8}\right)^{L-1} + o\big((3/8)^L\big) + \mathcal{O}(\delta^2).
	\end{equation*}
	This establishes that $\delta = \mathcal{O}((3/8)^L)$, which implies $\mathcal{O}(\delta^2) = o((3/8)^L)$. Thus, we obtain the rigorous asymptotic leading term for $\delta$:
	\begin{equation*}
		\delta = \frac{15}{4} \left( \frac{8}{3} \right)^{1-L} (1 + o(1)).
	\end{equation*}
	Substituting $\lambda = 16 - \delta$ into Eq.~\eqref{eq:r1_ineq} and applying the Taylor expansion $\ln(1 - x) = -x + \mathcal{O}(x^2)$ for $x \to 0$:
	\begin{align*}
		r_1 &\le 2(m-2) \frac{\ln(16 / (16-\delta))}{\ln 16} + 2 \nonumber \\
		&= -\frac{2(m-2)}{\ln 16} \ln \left( 1 - \frac{\delta}{16} \right) + 2 \nonumber \\
		&= \frac{2(m-2)}{\ln 16} \left( \frac{\delta}{16} + \mathcal{O}(\delta^2) \right) + 2 \nonumber \\
		&= \frac{m-2}{8 \ln 16} \left[ \frac{15}{4} \left( \frac{8}{3} \right)^{1-L} (1 + o(1)) \right] + 2.
	\end{align*}
	Since $L = \lceil \log_{8/3} m \rceil$, we have $(8/3)^{1-L} \le (8/3) \cdot (8/3)^{-\log_{8/3} m} = \frac{8}{3m}$. Substituting this upper bound into the inequality:
	\begin{align*}
		r_1 &\le \frac{m}{8 \ln 16} \cdot \frac{15}{4} \cdot \frac{8}{3m} (1 + o(1)) + 2 \\
		&= \frac{5}{4 \ln 16} + 2 + o(1) = \mathcal{O}(1).
	\end{align*}
	This indicates that the redundancy introduced by the generalized RLL encoding is bounded by a constant, independent of the code length $n$.

	\subsection{Analysis of Structural and Parity Redundancy ($r_2$)}\label{subsec:r2_cal}
	
	Next, we evaluate $r_2 = n - tm + 2$. This requires estimating the parameter $m$. By definition (Eq.~\eqref{eq:m_constraint}), $m$ is the maximum integer satisfying the constraint. Therefore, $m' \triangleq m+1$ violates the constraint, yielding:
	\begin{equation} \label{eq:m_prime_ineq}
		\left\lceil \frac{r_d' - (t-1)}{t} \right\rceil + m' + 3 > \frac{n}{t},
	\end{equation}
	where $r_d'$ is the parity length corresponding to the information bit length $m'$:
	\begin{align*}
		r_d' = 2 + \Biggl\lceil  t \log_4 22 - \log_4 2 + \log_4 m' + (t-1) \log_4(P'+1) \Biggr\rceil,
	\end{align*}
	and $P' = \lceil \log_{8/3} m' \rceil + 2$.
	
	First, we bound $r_d'$. Since $m \le n/t -3$, it follows that $m' < n/t$, giving:
	\begin{align}
		r_d' &\le 3 + t \log_4 22 + \log_4 \frac{n}{t} + (t-1) \log_4 \left( \log_{8/3} \frac{n}{t} + 4 \right) \nonumber \\
		&= \log_4 n + (t-1) \log_4 (\log_{8/3} n) + \mathcal{O}(1). \label{eq:rd_prime_est}
	\end{align}
	Rearranging Eq.~\eqref{eq:m_prime_ineq}:
	\begin{equation*}
		\frac{r_d' - (t-1)}{t} + m + 5 > \frac{n}{t} \implies m > \frac{n-1}{t} - \frac{r_d'}{t} - 4.
	\end{equation*}
	Substituting this into the definition of $r_2$:
	\begin{align*}
		r_2 &= n - t m + 2  \\
		&< n - t \left( \frac{n-1}{t} - \frac{r_d'}{t} - 4 \right) + 2  \\
		&= r_d' + 4t + 3.
	\end{align*}
	Substituting Eq.~\eqref{eq:rd_prime_est} into the above yields:
	\begin{equation*}
		r_2 < \log_4 n + (t-1) \log_4 (\log_{8/3} n) + \mathcal{O}(1).
	\end{equation*}
	
	\subsection{Total Redundancy and Asymptotic Optimality} \label{subsec:redundancy_optimality}
	
	Based on the preceding analysis of $r_1$ and $r_2$, we establish the total redundancy of the proposed scheme and evaluate its optimality against the theoretical fundamental limits.
	
	\begin{theorem}\label{thm:cons_redundancy}
		For a code length $n$ and a burst deletion/insertion length $t$ (satisfying $t \mid n$, $n \ge 7t + 3$, and $t \ge 2$), the  redundancy $r(n, t)$ of the proposed burst $t$-deletion/insertion $\mathcal{A}$-labeling code satisfies:
		\begin{equation*}
			r(n, t) = \log_4 n + (t-1) \log_4 (\log_{8/3} n) + \mathcal{O}(1).
		\end{equation*}
	\end{theorem}
	
	\begin{IEEEproof}
		From Eq.~\eqref{eq:r_r1_r2}, the total redundancy is $r = r_1 + r_2$. Section~\ref{subsec:r1_cal} shows that $r_1 = \mathcal{O}(1)$, and Section~\ref{subsec:r2_cal} establishes that $r_2 \le \log_4 n + (t-1) \log_4 (\log_{8/3} n) + \mathcal{O}(1)$. Hence the dominant term of $r(n,t)$ is $\log_4 n$, the sub-dominant term is $(t-1) \log_4 \log_{8/3} n$, and all remaining constants are absorbed into the $\mathcal{O}(1)$ term.
	\end{IEEEproof}
	
	Comparing with Theorem~\ref{thm:theoretical_bound_new}, which gives $r_{\text{lower}} = \log_4 n + \mathcal{O}(1)$, we see that the leading term $\log_4 n$ of $r(n,t)$ matches the lower bound. The gap $(t-1)\log_4\log_{8/3} n = \mathcal{O}(\log\log n)$ grows extremely slowly, so the proposed scheme is asymptotically optimal in $n$ for every fixed $t\ge 2$.
	\section{Conclusion} \label{sec:conclusion}
	
	We have studied burst deletion and insertion errors in label-based DNA data retrieval over the minimum-cardinality length-two label set $\mathcal{S}$. We formally introduced burst $t$-deletion/insertion $\mathcal{A}$-labeling codes and proved a redundancy lower bound of $\log_4 n + \mathcal{O}(1)$ for all $t\ge 1$ with $t\mid n$, settling a problem that was previously open even for $t=1$. We also gave an explicit construction whose encoder and decoder run in $\mathcal{O}(n^2)$ time and achieve redundancy $\log_4 n + (t-1)\log_4\log_{8/3} n + \mathcal{O}(1)$, matching the lower bound up to an $\mathcal{O}(\log\log n)$ gap.
	
	Several questions remain open. The most natural is whether the $(t-1)\log_4\log_{8/3} n$ overhead can be removed, yielding an exactly optimal construction. Beyond this, extending the framework to (i) multiple bursts, (ii) substitution--indel mixed bursts in the label domain, and (iii) other minimal label families would broaden its practical scope.
	
	\bibliographystyle{IEEEtran}
	\bibliography{ref}

    \appendices
    
    \section{Generalization of the Fundamental Limit to Arbitrary Constant Label Lengths}
\label{sec:appendix_generalization}

In this appendix, we generalize the fundamental information-theoretic lower bound on the redundancy of burst $t$-deletion $\mathcal{A}$-labeling codes (Theorem 4) to accommodate any label set $\mathcal{A}$ with a constant label length $l \ge 2$.

\begin{theorem}
Let $t \ge 1$ and $l \ge 2$ be fixed integers. Let $\mathcal{A}$ be a label set consisting of length-$l$ labels such that the original sequence can be uniquely reconstructed from its labeling sequence given the boundary bases. For $t \mid n$, the redundancy of any burst $t$-deletion $\mathcal{A}$-labeling code $\mathcal{C} \subseteq \Sigma_4^n$ satisfies
\begin{equation*}
    r(\mathcal{C}) \ge \log_4 n + \mathcal{O}(1).
\end{equation*}
\end{theorem}

\begin{IEEEproof}
Let $\mathcal{C} \subseteq \Sigma_4^n$ be a valid burst $t$-deletion $\mathcal{A}$-labeling code. For any $x \in \mathcal{C}$, let $y = L_{\mathcal{A}}(x) \in \Sigma_{|\mathcal{A}|+1}^n$ be its labeling sequence. A burst $t$-deletion shortens $y$ to length $n-t$. We denote the restricted labeling mapping defined on length-$(n-t)$ sequences by $L_{\mathcal{A}}': \Sigma_4^{n-t} \rightarrow \Sigma_{|\mathcal{A}|+1}^{n-t}$.

\textit{Step 1: The Reduced Received Space.}
Let the valid received space be $\mathcal{R}_{\text{valid}} \triangleq L_{\mathcal{A}}'(\Sigma_4^{n-t})$. Trivially, $|\mathcal{R}_{\text{valid}}| \le 4^{n-t}$. 
Define the effective deletion ball of $x$ as $\mathcal{B}_{\text{valid}}(x) \triangleq \mathcal{B}L_{\mathcal{A}}^{\text{del}}(x,t) \cap \mathcal{R}_{\text{valid}}$. Since $\mathcal{C}$ is a $t$-burst deletion correcting code, the effective balls are pairwise disjoint. Thus, we have the sphere-packing constraint:
\begin{equation} \label{eq:gen_sphere_packing}
    \sum_{x \in \mathcal{C}} |\mathcal{B}_{\text{valid}}(x)| \le |\mathcal{R}_{\text{valid}}| \le 4^{n-t}.
\end{equation}

\textit{Step 2: Constructing Independent Splice Centers.}
To establish a lower bound on $|\mathcal{B}_{\text{valid}}(x)|$, we partition the sequence $x$ into $W \triangleq \lfloor n / (t+l) \rfloor$ non-overlapping contiguous blocks of length $t+l$. 
For a block starting at index $k$, it spans the coordinates $x_{[k, k+t+l-1]}$. We define this block to be a \textit{valid splice center} if it simultaneously satisfies two conditions:
\begin{enumerate}
    \item \textbf{Splice Condition:} $x_j = x_{j+t}$ for all $k \le j \le k+l-2$. This enforces a local periodicity of $t$ for $l-1$ consecutive bases.
    \item \textbf{Mismatch Condition:} $x_{k+l-1} \neq x_{k+t+l-1}$. This explicitly breaks the periodicity immediately after the splice region.
\end{enumerate}
Let $N_{\text{splice}}(x)$ denote the number of valid splice centers in $x$. We claim that $|\mathcal{B}_{\text{valid}}(x)| \ge N_{\text{splice}}(x)$. 

To prove this, we construct an injection from the set of valid splice centers to $\mathcal{B}_{\text{valid}}(x)$. For a valid splice center starting at $k$, let $x_{\text{del}}(k) \in \Sigma_4^{n-t}$ be the sequence obtained by deleting $x_{[k, k+t-1]}$. 
Because the splice condition ensures $x_j = x_{j+t}$ for the $l-1$ symbols bridging the gap, every length-$l$ window in $x_{\text{del}}(k)$ strictly preserves the dinucleotides (or $l$-mers) present in $x$. Consequently, $L_{\mathcal{A}}'(x_{\text{del}}(k))$ is exactly the sequence obtained by deleting a burst of $t$ labels starting at index $k$ in $y$. Therefore, $L_{\mathcal{A}}'(x_{\text{del}}(k)) \in \mathcal{B}_{\text{valid}}(x)$.

Now, consider two distinct valid splice centers starting at $k$ and $k'$, with $k < k'$. Since the blocks are non-overlapping, $k' \ge k + t + l$. We compare $x_{\text{del}}(k)$ and $x_{\text{del}}(k')$ at the specific coordinate $k+l-1$:
\begin{itemize}
    \item In $x_{\text{del}}(k)$, the deletion of $x_{[k, k+t-1]}$ shifts the original coordinate $k+t+l-1$ leftward by $t$. Hence, the symbol at index $k+l-1$ in $x_{\text{del}}(k)$ is the original $x_{k+t+l-1}$.
    \item In $x_{\text{del}}(k')$, since the deletion occurs at $k' > k+l-1$, no shift occurs at index $k+l-1$. The symbol at index $k+l-1$ in $x_{\text{del}}(k')$ remains the original $x_{k+l-1}$.
\end{itemize}
By the Mismatch Condition of the block at $k$, $x_{k+l-1} \neq x_{k+t+l-1}$. This guarantees that $x_{\text{del}}(k) \neq x_{\text{del}}(k')$. 
Furthermore, since both length-$(n-t)$ sequences share identical boundary bases, the unique reconstruction property of $\mathcal{A}$ ensures that $L_{\mathcal{A}}'(x_{\text{del}}(k)) \neq L_{\mathcal{A}}'(x_{\text{del}}(k'))$. This confirms the injection, yielding $|\mathcal{B}_{\text{valid}}(x)| \ge N_{\text{splice}}(x)$.

\textit{Step 3: Bounding the Low-Density Subset.}
We now evaluate the probability that a random block of length $t+l$ constitutes a valid splice center. Over the $4^{t+l}$ possible assignments, the first $t$ bases are unconstrained ($4^t$ choices). The Splice Condition deterministically fixes the next $l-1$ bases. The Mismatch Condition offers $3$ choices for the final base (avoiding equality with the corresponding base $t$ positions prior). Thus, the number of valid assignments is $4^t \times 1 \times 3 = 3 \cdot 4^t$.
The probability of forming a valid splice center is $p = \frac{3 \cdot 4^t}{4^{t+l}} = 3 \cdot 4^{-l}$. Crucially, $p$ is a positive constant strictly independent of $n$.

Let $\delta = p/2$, and define the threshold $\rho = \lfloor \delta W \rfloor$. We partition $\mathcal{C}$ into $\mathcal{C}_{\text{low}} \triangleq \{x \in \mathcal{C} : N_{\text{splice}}(x) \le \rho\}$ and $\mathcal{C}_{\text{high}} \triangleq \mathcal{C} \setminus \mathcal{C}_{\text{low}}$.

Since the $W$ blocks are non-overlapping and structurally independent in the unconstrained space $\Sigma_4^n$, we can establish an upper bound on $V_{\text{low}}$ by counting the number of sequences that contain exactly $j$ valid splice centers ($0 \le j \le \rho$) strictly within these $W$ designated blocks. The counting process decomposes into four independent factors:

\begin{enumerate}
    \item \textbf{The Unconstrained Tail:} The sequence is partitioned into $W$ blocks of length $t+l$, leaving a tail of length $n - W(t+l)$. Since we place no restrictions on this tail, it contributes exactly $4^{n - W(t+l)}$ possible combinations.
    \item \textbf{Block Selection:} There are $\binom{W}{j}$ ways to choose exactly $j$ blocks out of the $W$ available blocks to act as the valid splice centers.
    \item \textbf{Valid Blocks:} As derived previously, a single block of length $t+l$ has exactly $3 \cdot 4^t$ sequences that satisfy both the Splice and Mismatch conditions. The $j$ valid blocks therefore contribute a factor of $(3 \cdot 4^t)^j$.
    \item \textbf{Invalid Blocks:} The remaining $W-j$ blocks must strictly \textit{not} be valid splice centers. A block of length $t+l$ has $4^{t+l}$ total possible quaternary sequences. Subtracting the valid ones leaves $4^{t+l} - 3 \cdot 4^t$ invalid sequences per block, contributing a factor of $(4^{t+l} - 3 \cdot 4^t)^{W-j}$.
\end{enumerate}

Any sequence with at most $\rho$ valid splice centers in total across the entire length $n$ must, by definition, have at most $\rho$ valid splice centers localized within our specific $W$ disjoint blocks. Therefore, multiplying these independent combinatorial choices and summing over all possible values of $j$ from $0$ to $\rho$ provides a rigorous upper bound for $V_{\text{low}}$:
\begin{equation*}
    V_{\text{low}} \le \sum_{j=0}^{\rho} \binom{W}{j} (3 \cdot 4^t)^j (4^{t+l} - 3 \cdot 4^t)^{W-j} \cdot 4^{n - W(t+l)}.
\end{equation*}
Recall that the probability of a block being a valid splice center is $p = \frac{3 \cdot 4^t}{4^{t+l}}$. We can factor out $(4^{t+l})^W$ from the summation:
\begin{align*}
    V_{\text{low}}
    \le& 4^{n - W(t+l)} (4^{t+l})^W \sum_{j=0}^{\rho} \binom{W}{j} \left(\frac{3 \cdot 4^t}{4^{t+l}}\right)^j \left(1 - \frac{3 \cdot 4^t}{4^{t+l}}\right)^{W-j} \\
    =& 4^n \sum_{j=0}^{\rho} \binom{W}{j} p^j (1-p)^{W-j}.
\end{align*}
Notice that the summation term is exactly the cumulative distribution function $\mathbb{P}(S \le \rho)$ of a binomial random variable $S \sim \text{Binomial}(W, p)$~\cite{Cover2005}. The expected value is $\mu = \mathbb{E}[S] = pW$. Since $\delta = p/2$, our threshold is $\rho = \lfloor \frac{p}{2} W \rfloor \le \frac{1}{2}\mu$. 

We can now apply the Chernoff bound for the lower tail of a binomial distribution~\cite{Cover2005}, which states that for any $0 < \epsilon < 1$, $\mathbb{P}(S \le (1-\epsilon)\mu) \le \exp\left(-\frac{\epsilon^2 \mu}{2}\right)$. Setting $\epsilon = 1/2$, we obtain:
\begin{equation*}
    \mathbb{P}\left(S \le \frac{1}{2}pW\right) \le \exp\left(-\frac{(1/2)^2 pW}{2}\right) = \exp\left(-\frac{p}{8} W\right).
\end{equation*}
Since $W = \lfloor \frac{n}{t+l} \rfloor > \frac{n}{t+l} - 1$, we can further bound the exponential decay. Let $\gamma \triangleq \frac{p}{8(t+l)}$, which is a strictly positive constant because $p$, $t$, and $l$ are all independent of $n$. We have:
\begin{equation*}
    \exp\left(-\frac{p}{8} W\right) < \exp\left(-\frac{p}{8} \left(\frac{n}{t+l} - 1\right)\right) = \exp\left(\frac{p}{8}\right) e^{-\gamma n}.
\end{equation*}
Substituting this back into the bound for $V_{\text{low}}$, we get $V_{\text{low}} < 4^n \exp(p/8) e^{-\gamma n}$. To show that this term is negligible compared to the received space size, we examine the ratio:
\begin{equation*}
    \frac{V_{\text{low}}}{4^{n-t} / n^2} < \frac{4^n \exp(p/8) e^{-\gamma n}}{4^{n-t} n^{-2}} = 4^t \exp(p/8) \cdot n^2 e^{-\gamma n}.
\end{equation*}
Because exponential decay strictly dominates any polynomial growth asymptotically, we have $\lim_{n \to \infty} n^2 e^{-\gamma n} = 0$. Consequently, $V_{\text{low}} = o(4^{n-t}/n^2)$. 

Since the codebook $\mathcal{C}$ is a subset of $\Sigma_4^n$, the number of codewords falling into the low-density category must be bounded by the total number of such sequences in the entire space. Thus, $|\mathcal{C}_{\text{low}}| \le V_{\text{low}}$, immediately implying:
\begin{equation*}
    |\mathcal{C}_{\text{low}}| = o\left(\frac{4^{n-t}}{n^2}\right).
\end{equation*}

\textit{Step 4: Sphere-Packing for the High-Density Subset.}
For every $x \in \mathcal{C}_{\text{high}}$, we have $|\mathcal{B}_{\text{valid}}(x)| \ge \rho = \Omega(n)$. Substituting this into Eq.~\eqref{eq:gen_sphere_packing}:
\begin{equation*}
    |\mathcal{C}_{\text{high}}| \cdot \Omega(n) \le \sum_{x \in \mathcal{C}_{\text{high}}} |\mathcal{B}_{\text{valid}}(x)| \le 4^{n-t}.
\end{equation*}
This yields $|\mathcal{C}_{\text{high}}| = \mathcal{O}(4^{n-t}/n)$.

\textit{Step 5: Conclusion.}
Combining both subsets, the total cardinality is bounded by:
\begin{equation*}
    |\mathcal{C}| = |\mathcal{C}_{\text{low}}| + |\mathcal{C}_{\text{high}}| = \mathcal{O}\left(\frac{4^{n-t}}{n}\right).
\end{equation*}
Consequently, the redundancy evaluates to:
\begin{equation*}
    r(\mathcal{C}) = n - \log_4 |\mathcal{C}| \ge \log_4 n + \mathcal{O}(1),
\end{equation*}
which matches the dominant term of the original theorem, completing the proof.
\end{IEEEproof}
\end{document}